\documentclass[journal,twocolumn,twoside]{IEEEtran}

\usepackage{amsfonts}
\usepackage{color}
\usepackage{graphicx}
\usepackage[dvips]{epsfig}
\usepackage{graphics} 
\usepackage{times} 
\usepackage[cmex10]{amsmath} 
\usepackage{amssymb}  
\usepackage{cite}
\usepackage{multirow}
\usepackage[tight,footnotesize]{subfigure}
\usepackage{algorithm}
\usepackage{amsthm}
\usepackage{mathtools}
\usepackage{cuted}
\usepackage{stfloats}

\usepackage{multicol}

\def\norm #1{\left\|#1\right\|}

\def\twon #1{\left\|#1\right\|_2}
\def\onen #1{\left\|#1\right\|_1}

\def\frobn #1{\left\|#1\right\|_{\text{F}}}

\def\atomn #1{\left\|#1\right\|_{\cA}}

\def\datomn #1{\left\|#1\right\|_{\cA}^*}

\def\abs #1{\left|#1\right|}
\def\inp #1{\left\langle#1\right\rangle}

\def\st{\text{subject to }}

\def\bC{\mathbb{C}}

\def\bR{\mathbb{R}}
\def\bE{\mathbb{E}}
\def\bP{\mathbb{P}}
\def\bT{\mathbb{T}}

\def\bS{\mathbb{S}}

\def\m #1{\boldsymbol{#1}}

\def\cA{\mathcal{A}}

\def\cE{\mathcal{E}}
\def\cF{\mathcal{F}}

\def\cJ{\mathcal{J}}
\def\cK{\mathcal{K}}

\def\cN{\mathcal{N}}

\def\cT{\mathcal{T}}

\def\ccn{\mathcal{CN}}

\def\bee{\begin{equation}}
\def\ene{\end{equation}}

\def\beq{\begin{eqnarray}}
\def\enq{\end{eqnarray}}
\def\lentwo{\setlength\arraycolsep{2pt}}

\newtheorem{lem}{Lemma}

\newtheorem{cor}{Corollary}
\newtheorem{thm}{Theorem}
\newtheorem{prop}{Proposition}

\def\equ #1{\begin{equation}#1\end{equation}}
\def\equa #1{\begin{eqnarray}#1\end{eqnarray}}
\def\sbra #1{\left(#1\right)}
\def\mbra #1{\left[#1\right]}
\def\lbra #1{\left\{#1\right\}}

\def\tr #1{\text{tr}#1}

\def\rank #1{\text{rank}#1}
\def\st {\text{ subject to }}

\title{On the Sample Complexity of Multichannel Frequency Estimation via Convex Optimization}
\author{Zai Yang, {\em Member, IEEE}, Jinhui Tang, {\em Senior Member, IEEE}, Yonina C.~Eldar, {\em Fellow, IEEE},\\ and Lihua Xie, {\em Fellow, IEEE}
\thanks{Part of this work will be presented at the 23rd IEEE International Conference on Digital Signal Processing (DSP 2018), Shanghai, China, November 2018 \cite{yang2018average}.

Z.~Yang is with the School of Automation, Nanjing University of Science and Technology, Nanjing 210094, China (e-mail: yangzai@njust.edu.cn).

J.~Tang is with the School of Computer Science and Engineering, Nanjing University of Science and Technology, Nanjing 210094, China (e-mail: jinhuitang@njust.edu.cn).

Y.~C.~Eldar is with the Department of Electrical Engineering, Technion--Israel Institute of Technology, Haifa 32000, Israel (e-mail: yonina@ee.technion.ac.il).

L.~Xie is with the School of Electrical and Electronic Engineering, Nanyang Technological University, Singapore 639798 (e-mail: elhxie@ntu.edu.sg).}}

\begin{document}
\maketitle

\begin{abstract}
The use of multichannel data in line spectral estimation (or frequency estimation) is common for improving the estimation accuracy in array processing, structural health monitoring, wireless communications, and more. Recently proposed atomic norm methods have attracted considerable attention due to their provable superiority in accuracy, flexibility and robustness compared with conventional approaches. In this paper, we analyze atomic norm minimization for multichannel frequency estimation from noiseless compressive data, showing that the sample size per channel that ensures exact estimation decreases with the increase of the number of channels under mild conditions. In particular, given $L$ channels, order $K\sbra{\log K} \sbra{1+\frac{1}{L}\log N}$ samples per channel, selected randomly from $N$ equispaced samples, suffice to ensure with high probability exact estimation of $K$ frequencies that are normalized and mutually separated by at least $\frac{4}{N}$. Numerical results are provided corroborating our analysis.
\end{abstract}


\section{Introduction}
%

Line spectral estimation or frequency estimation is a fundamental problem in statistical signal processing\cite{stoica2005spectral}. It refers to the process of estimating the frequency and amplitude parameters of several complex sinusoidal waves from samples of their superposition. In this paper, we consider a compressive multichannel setting in which each channel observes a subset of equispaced samples and the sinusoids among the multiple channels share the same frequency profile.

The full $N\times L$ data matrix $\m{Y}^o$ is composed of equispaced samples $\lbra{y^o_{jl}}$ that are given by
\equ{y_{jl}^o = \sum_{k=1}^K e^{i2\pi(j-1)f_k} s_{kl}, \quad j=1,\dots,N, \; l = 1,\dots,L, \label{eq:model}}
where $N$ is the full sample size per channel, $L$ is the number of channels, $i = \sqrt{-1}$, and $s_{kl}$ denotes the unknown complex amplitude of the $k$th frequency in the $l$th channel. We assume that the frequencies $\lbra{f_k}$ are normalized by the sampling rate and belong to the unit circle $\bT = \mbra{0,1}$, where $0$ and $1$ are identified. Letting $\m{a}\sbra{f} = \mbra{1,e^{i2\pi f},\dots,e^{i2\pi(N-1)f}}^T$ represent a complex sinusoid of frequency $f$, and $\m{s}_k = \mbra{s_{k1},\dots,s_{kL}}$ the coefficient vector of the $k$th component, we can express the full data matrix $\m{Y}^o$ as
\equ{\m{Y}^o = \sum_{k=1}^K \m{a}\sbra{f_k}\m{s}_k. \label{eq:model2}}
In the absence of noise, the goal of compressive multichannel frequency estimation is to estimate the frequencies $\lbra{f_k}$ and the amplitudes $\lbra{s_{kl}}$, given a subset of the rows of $\m{Y}^o$. In the case of $L=1$ this problem is referred to as continuous/off-the-grid compressed sensing \cite{tang2012compressed}, or spectral super-resolution by swapping the roles of frequency and time (or space) \cite{candes2013towards}.

Multichannel frequency estimation appears in many applications such as array processing\cite{krim1996two,stoica2005spectral}, structural health monitoring\cite{heylen2006modal}, wireless communications\cite{barbotin2012estimation}, radar\cite{li2007mimo,baransky2014sub,eldar2015sampling}, and fluorescence microscopy \cite{rust2006sub}. In array processing, for example, one needs to estimate the directions of several electromagnetic sources from outputs of several antennas that form an antenna array. In the common uniform-linear-array case, the rows of $\m{Y}^o$ correspond to (spatial) locations of the antennas, each column of $\m{Y}^o$ corresponds to one (temporal) snapshot, and each frequency has a one-to-one mapping to the direction of one source. Multichannel data are available if multiple data snapshots are acquired. The compressive data case arises when the antennas form a sparse linear array \cite{rossi2014spatial}. In structural health monitoring, a number of sensors are deployed to collect vibration data of a physical structure (e.g., a bridge or building) that are then used to estimate the structure's modal parameters such as natural frequencies and mode shapes. In this case, the rows of $\m{Y}^o$ correspond to the (temporal) sampling instants and each column to one (spatially) individual sensor. Multichannel data can be acquired by deploying multiple sensors. The use of compressive data has potential to relax the requirements on data acquisition, storage and transmission.

Due to its connection to array processing, multichannel frequency estimation has a long history dating back to at least World War II when the Bartlett beamformer was proposed to estimate the directions of enemy planes. The use of multichannel data can improve estimation performance by taking more samples and exploiting redundancy among the data channels. As the number of channels increases, the Fisher information about the frequency parameters increases if the data are independent among channels \cite{schervish2012theory}, and the number of identifiable frequencies increases in general as well \cite{wax1989unique}. However, in the worst case where all vector coefficients $\lbra{\m{s}_k}$ are identical up to scaling factors, referred to as coherent sources in array processing, the multichannel data will be identical up to scaling factors. Therefore, the frequency estimation performance cannot be improved by increasing the number of channels. The possible presence of coherent sources is known to form a difficult scenario for multichannel frequency estimation.

In multichannel frequency estimation, the sample data is a highly nonlinear function of the frequencies. Consequently, (stochastic and deterministic) maximum likelihood estimators \cite{krim1996two} are difficult to compute. Subspace-based methods were proposed to circumvent nonconvex optimization and dominated research on this topic for several decades. However, these techniques have drawbacks such as sensitivity to coherent sources and limited applicability in the case of missing data or when the number of channels is small \cite{krim1996two,stoica2005spectral}. Convex optimization approaches have recently been considered and are very attractive since they can be globally solved in polynomial time. In addition, they are applicable to various kinds of noise and data structures, are robust to coherent sources and have provable performance guarantees \cite{candes2013towards,candes2013super,tang2012compressed, azais2015spike, tang2015near,yang2016exact,li2016off, fernandez2017demixing, li2018atomic}. However, the performance of such techniques under multichannel data has not been thoroughly derived. Extensive numerical results presented in \cite{yang2016exact,li2016off,fernandez2016super} show that the frequency estimation performance improves (in terms of accuracy and resolution) with the number of channels in the case of noncoherent sources; nevertheless, few theoretical results have been developed to support this empirical behavior.

In this paper, we provide a rigorous analysis of one of the empirical findings: The sample size per channel that ensures exact frequency estimation from noiseless data can be reduced as the number of channels increases. Our main result is described in the following theorem.

\begin{thm} Suppose we observe the $N\times L$ data matrix $\m{Y}^o = \sum_{k=1}^K \m{a}\sbra{f_k}\m{s}_k$ on rows indexed by $\Omega\subset\lbra{1,\dots,N}$, where $\Omega$ is of size $M$, selected uniformly at random, and given. Assume that the vector phases $\m{\phi}_k = \frac{\m{s}_k}{\twon{\m{s}_k}}$ are independent random variables that are drawn uniformly from the unit $L$-dimensional complex sphere, and that the frequency support $\cT = \lbra{f_k}$ satisfies the minimum separation condition
\equ{\Delta_{\cT} \coloneqq \min_{p\neq q} \abs{f_p-f_q} > \frac{1}{\left\lfloor (N-1)/4 \right\rfloor}, \label{eq:separation}}
where the distance is wrapped around on the unit circle.
Then, there exists a numerical constant $C$, which is fixed and independent of the parameters $K$, $N$, $L$ and $\delta$, such that
\equ{M\geq C\max\lbra{\log^2\frac{N}{\delta}, K\sbra{\log\frac{K}{\delta}}\sbra{1+\frac{1}{L}\log\frac{N}{\delta}}} \label{eq:AN_bound}}
is sufficient to guarantee that, with probability at least $1-\delta$, the full data $\m{Y}^o$ and its frequency support $\cT$ can be uniquely produced by solving a convex optimization problem (defined in \eqref{eq:SDP_ANM} or \eqref{eq:SDP_dual} below). \label{thm:noiseless}
\end{thm}

Theorem \ref{thm:noiseless} provides a non-asymptotic analysis for multichannel frequency estimation that holds for arbitrary number of channels $L$. According to \eqref{eq:AN_bound}, the sample size per channel that ensures exact frequency estimation decreases as $L$ increases. Consider the case of interest in which the lower bound is dominated by the second term in it, which is true once the number of frequencies is modestly large such that $K\log{\frac{K}{\delta}}\geq \min\lbra{\log^2 \frac{N}{\delta},\; L\log \frac{N}{\delta}}$. The sample size per channel then becomes
\equ{M\geq C K\sbra{\log\frac{K}{\delta}} \sbra{1+\frac{1}{L}\log\frac{N}{\delta}}, \label{eq:Mbound2}}
where the lower bound is a monotonically decreasing function of $L$. In fact, to solve for $\lbra{f_k}$ and $\lbra{s_{kl}}$ which consist of $(2L+1)K$ real unknowns, the minimum number of complex-valued samples is $\frac{1}{2}(2L+1)K$. Consequently, the sample size per channel for any method must satisfy
\equ{M\geq \frac{1}{2L}(2L+1)K = K\sbra{1+ \frac{1}{2L}}.}
Our sample order in \eqref{eq:Mbound2} nearly reaches this information-theoretic rate up to logarithmic factors of $K$ and $N$. In addition, they share an identical decreasing rate with respect to $L$.

The assumption on the vector phases $\m{\phi}_k = \frac{\m{s}_k}{\twon{\m{s}_k}}$ is made in Theorem \ref{thm:noiseless} to avoid coherent sources. It is inspired by \cite{gribonval2008atoms,eldar2010average} which show performance gains of multiple channels in the context of sparse recovery. This assumption is satisfied by the rotational invariance of Gaussian random vectors if $\lbra{s_{kl}}$ are independent complex Gaussian random variables with
\equ{s_{kl} \sim \ccn\sbra{0, p_k},\notag}
where the variance $p_k$ refers to the power of the $k$th source. This latter assumption is usually adopted in the literature to derive the Cramer-Rao bound or to analyze the theoretical performance of subspace-based methods \cite{krim1996two}.

Note that the $4/N$ minimum separation condition \eqref{eq:separation} comes from the single-channel analysis in \cite{candes2013towards,tang2012compressed}. It is reasonably tight in the sense that, if the separation is under $2/N$, then the data corresponding to two different sets of frequencies can be practically the same and it becomes nearly impossible to distinguish them \cite{fernandez2016super}. This condition is conservative in practice, and a numerical approach was implemented in \cite{yang2016enhancing} to further improve resolution. Such a separation condition is not required for subspace-based methods, at least in theory.

A corollary of Theorem \ref{thm:noiseless} can be obtained when the frequencies lie on a known grid. This scenario may be unrealistic but is often assumed in practice to simplify the algorithm design.

\begin{cor} Consider the problem setup of Theorem \ref{thm:noiseless} when in addition the frequencies $\lbra{f_k}$ lie on a given fixed grid (that can be nonuniform or arbitrarily fine). Under the assumptions of Theorem \ref{thm:noiseless} and given the same number of samples, the frequencies can be exactly estimated with at least the same probability by solving an $\ell_{2,1}$ norm minimization problem (defined in \eqref{eq:l21} below). \label{cor:l21}
\end{cor}

It is worth noting that Corollary \ref{cor:l21} holds regardless of the density of the grid, while the grid is used to define the $\ell_{2,1}$ norm minimization problem. In fact, Theorem \ref{thm:noiseless} corresponds to the extreme case of Corollary \ref{cor:l21} in which the grid becomes infinitely fine and contains infinitely many points. In this case, the $\ell_{2,1}$ norm minimization problem in Corollary \ref{cor:l21} becomes computationally intractable. In contrast, the convex optimization problem in Theorem \ref{thm:noiseless} is still applicable.

The rest of the paper is organized as follows. Section \ref{sec:priorart} revisits prior art in multichannel frequency estimation and related topics and discusses their connections to this work. Section \ref{sec:ANM} presents the convex optimization methods mentioned in Theorem \ref{thm:noiseless} and Corollary \ref{cor:l21}. The former method is referred to as atomic norm minimization originally introduced in \cite{yang2016exact,li2016off}. The latter is known as $\ell_{2,1}$ norm minimization \cite{malioutov2005sparse,eldar2010average}. Section \ref{sec:simulation} presents simulation results to verify our main theorem. Section \ref{sec:proof} provides the detailed proof of Theorem \ref{thm:noiseless}. Finally, Section \ref{sec:conclusion} concludes this paper.

Throughout the paper, $\log a\,b$ means $\sbra{\log a} b$ unless we write $\log\sbra{ab}$. The notations $\bR$ and $\bC$ denote the set of real and complex numbers, respectively. The normalized frequency interval is identified with the unit circle $\bT=\mbra{0,1}$. Boldface letters are reserved for vectors and matrices. The
$\ell_1$, $\ell_2$ and Frobenius norms are written as $\onen{\cdot}$, $\twon{\cdot}$ and $\frobn{\cdot}$ respectively, and $\abs{\cdot}$ is the magnitude of a scalar or the cardinality of a set. For matrix $\m{A}$, $\m{A}^T$ is the transpose, $\m{A}^H$ is the conjugate transpose, $\rank\sbra{\m{A}}$ denotes the rank, and $\tr\sbra{\m{A}}$ is the trace. The fact that a square matrix $\m{A}$ is positive semidefinite is expressed as $\m{A}\geq \m{0}$. Unless otherwise stated, $x_j$ is the $j$th entry of a vector $\m{x}$, $A_{jl}$ is the $(j,l)$th entry of a matrix $\m{A}$, and $\m{A}_j$ denotes the $j$th row. Given an index set $\Omega$, $\m{x}_{\Omega}$ and $\m{A}_{\Omega}$ are the subvector and submatrix of $\m{x}$ and $\m{A}$, respectively, that are formed by the rows of $\m{x}$ and $\m{A}$ indexed by $\Omega$. For a complex argument, $\Re$ and $\Im$ return its real and imaginary parts. The big O notation is written as $O\sbra{\cdot}$. The expectation of a random variable is denoted by $\bE\mbra{\cdot}$, and $\bP\sbra{\cdot}$ is the probability of an event.

\section{Prior Art} \label{sec:priorart}
The advantages of using multichannel data have been well understood for parametric approaches for frequency estimation, e.g., maximum likelihood and subspace-based methods. Under proper assumptions on $\lbra{s_{kl}}$, which are similar to those in Theorem \ref{thm:noiseless}, these techniques are asymptotically efficient in the number of channels $L$ and therefore their accuracy improves with the increase of $L$ \cite{stoica1989music}. However, few theoretical results on their accuracy are available with fixed $L$, in contrast to Theorem \ref{thm:noiseless} which holds for any $L$. Note that theoretical guarantees of two common subspace-based methods, MUSIC \cite{schmidt1981signal} and ESPRIT \cite{roy1989esprit}, have recently been developed in \cite{liao2016music,fannjiang2016compressive} in the single-channel full data case. For multichannel sparse recovery in discrete models, MUSIC has been applied or incorporated into other methods, with theoretical guarantees, see \cite{feng1996spectrum,davies2012rank, kim2012compressive,lee2012subspace}.

Theorem \ref{thm:noiseless} is related to the so-called average-case analysis for multichannel sparse recovery in compressed sensing problems \cite{gribonval2008atoms,eldar2010average,eldar2012compressed}. In such settings, the goal is to recover multiple sparse signals sharing the same support from their linear measurements. Average case analysis (as opposed to the worst case) attempts to analyze an algorithm's average performance for generic signals. Under an assumption on the sparse signals similar to that on $\lbra{s_{kl}}$ in Theorem \ref{thm:noiseless}, the papers \cite{gribonval2008atoms,eldar2010average} showed that, if the columns of the sensing matrix are weakly correlated (in terms of mutual coherence \cite{donoho2001uncertainty} or the restricted isometry property (RIP) \cite{candes2006compressive}), then the probability of successful sparse recovery improves as the number of channels increases. While \cite{gribonval2008atoms} used the greedy method of orthogonal matching pursuit (OMP) \cite{pati1993orthogonal,tropp2006algorithms}, the work in \cite{eldar2010average} was focused on $\ell_{2,1}$ norm minimization. The proof of Theorem \ref{thm:noiseless} is partially inspired by these results. However, the frequency parameters are continuously valued so that the weak correlation assumption is not satisfied. In addition to this, Theorem \ref{thm:noiseless} provides an explicit expression of the sample size per channel as a decreasing function of the number of channels.

Multichannel sparse recovery methods have been applied to multichannel frequency estimation using, e.g., $\ell_{2,1}$ norm minimization as in Corollary \ref{cor:l21} \cite{malioutov2005sparse}. These approaches utilize the fact that the number of frequencies is small, and attempt to find the smallest set of frequencies on a grid describing the observed data. It was empirically shown in a large number of publications (see, e.g., \cite{gorodnitsky1997sparse, malioutov2005sparse,stoica2011spice}) that, compared with conventional nonparametric and parametric methods, such sparse techniques have improved resolution and relaxed requirements on the number of channels and the type of sources and noise. However, these approaches can be applied only if the continuous frequency domain is discretized so that the frequencies are restricted to a finite set of grid points. Due to discretization errors and high correlations among the candidate frequency components, few theoretical guarantees were developed to support the good empirical results. Corollary \ref{cor:l21} provides theoretical support under the assumption of no discretization errors.

Assume in Corollary \ref{cor:l21} a uniform grid of size equal to $N$ (without taking into account discretization errors). The case of $L=1$ then degenerates to the compressed sensing problem studied in \cite{candes2006robust}, where the weak correlation assumption can be satisfied with sample complexity $O\sbra{K\log^2 K\log N}$ \cite{haviv2017restricted}. This means that given approximately the same number of samples, the resolution can be increased from $4/N$ to $1/N$ if a stricter assumption is made on the frequencies. In this case, the benefits of using multichannel data have been shown in \cite{gribonval2008atoms,eldar2010average}, as discussed previously.

The gap between the discrete and continuous frequency setups has been closed in the past few years. In the pioneering work of Cand{\`e}s and Fernandez-Granda \cite{candes2013towards}, the authors studied the single-channel full data setting and proposed a convex optimization method (to be specific, semidefinite program (SDP)) in the continuous domain, referred to as total variation minimization, that is a continuous analog of $\ell_1$ minimization in the discrete setup. They proved that the true frequencies can be recovered if they are mutually separated by at least $4/N$---the minimum separation condition \eqref{eq:separation}. The idea and method were then extended by Tang {\em et al.} \cite{tang2012compressed} to compressive data using atomic norm minimization \cite{chandrasekaran2012convex} which is equivalent to the total variation norm. They showed that exact frequency localization occurs with high probability under the $4/N$ separation condition if $M\geq O\sbra{K\log K \log N}$ samples are given. This result is a special case of Theorem \ref{thm:noiseless} in the single-channel case.

Multichannel frequency estimation using atomic norm was studied by Yang and Xie \cite{yang2016exact}, Li and Chi \cite{li2016off}, and Li {\em et al.} \cite{li2018atomic}. The paper \cite{yang2016exact} showed that exact frequency estimation occurs under the same $4/N$ separation condition with full data. In the compressive setting, a sample complexity of $O(K\log K \log (\sqrt{L}N))$ per channel was derived under the assumption of centered independent random phases $\lbra{\m{\phi}_k}$. Compared with Theorem \ref{thm:noiseless}, the assumption on $\lbra{\m{\phi}_k}$ is weaker but the sample complexity is larger. Unlike Theorem \ref{thm:noiseless}, the sources are allowed to be coherent in \cite{yang2016exact} and therefore the above results do not shed any light on the advantage of multichannel data. These results may be considered as a worst-case analysis, if we refer to Theorem \ref{thm:noiseless} as an average-case analysis.

A slightly different compressive setting was considered in \cite{li2016off}, in which the samples are taken uniformly at random from the entries (as opposed to the rows) of $\m{Y}^o$, and an average per-channel sample complexity of $O\sbra{K\log (LK) \log (LN)}$ was stated under an assumption on $\lbra{s_{kl}}$ similar to that in Theorem \ref{thm:noiseless}. This result is weaker than that of Theorem \ref{thm:noiseless} and again fails to show the advantage of multichannel data.\footnote{Note that an explicit proof of this result was not included in \cite{li2016off}, which was also pointed out in \cite[Footnote 8]{li2018atomic}.}


A different means for data compression was considered in \cite{li2018atomic}, in which each column of $\m{Y}^o$ is compressed by using a distinct sensing matrix generated from a Gaussian distribution. In this case $M\geq \max\lbra{\frac{8}{L}\log\frac{1}{\delta}+2, \; CK\log N}$ samples per channel suffice to guarantee exact frequency estimation with probability at least $1-\delta$ by using atomic norm minimization, where $C$ is a constant. The sample size $M$ decreases with the number of channels if the lower bound is dominated by the first term or equivalently if $L\sbra{CK\log N-2}$ is less than a constant depending on $\delta$. Evidently, such a benefit from multichannel data diminishes if any one of $L$, $K$ or $N$ becomes large. This benefit is a result of the different sampling matrices applied on each channel, which we do not require in our analysis.

Another convex optimization method operating on the continuum for multichannel frequency estimation is gridless SPICE (GLS) \cite{yang2014discretization,yang2015gridless, stoica2011spice}. It was shown to be equivalent to an atomic norm method for small $L$ and a weighted atomic norm method for large $L$ \cite{yang2018sparse,yang2016enhancing}. Under an assumption on $\lbra{s_{kl}}$ similar to that in Theorem \ref{thm:noiseless}, GLS is asymptotically efficient in $L$. However, such benefits from multichannel data have not been analyzed for small $L$.



Finally, sparse estimation methods using nonconvex optimization have been proposed in \cite{fang2016super, wu2018two}, but few theoretical guarantees on their estimation accuracy have been derived. Readers are referred to \cite{yang2018sparse} for a review on sparse methods for multichannel frequency estimation.

\section{Multichannel Atomic Norm Minimization} \label{sec:ANM}
Atomic norm \cite{chandrasekaran2012convex} provides a generic approach to finding a sparse representation of a signal by exploiting particular structures in it. It generalizes the $\ell_1$ norm for sparse signal recovery and the nuclear norm for low rank matrix recovery. For multichannel frequency estimation, the set of atoms is defined as \cite{yang2016exact,li2016off}:
\equ{\cA \coloneqq \lbra{\m{a}\sbra{f}\m{\phi}:\; f\in\bT, \; \m{\phi}\in\bS^{2L-1}}, \notag}
where
\equ{\bS^{2L-1} \coloneqq \lbra{\m{\phi}\in\bC^{1\times L}:\; \twon{\m{\phi}}=1} \notag}
denotes the unit $L$-dimensional complex sphere, or equivalently, the unit $2L$-dimensional real sphere. The atomic norm of a multichannel signal $\m{Y}\in\bC^{N\times L}$ is defined as the gauge function of the convex hull of $\cA$:
\equ{\begin{split}\atomn{\m{Y}}
&= \inf \lbra{t>0:\; \m{Y} \in t\text{conv}\sbra{\cA}}\\
&= \inf \left\{\sum_k c_k:\; \m{Y} = \sum_k c_k \m{a}\sbra{f_k}\m{\phi}_k, \;  c_k >0\right\} \\
&= \inf \lbra{\sum_k \twon{\m{s}_k}:\; \m{Y} = \sum_k \m{a}\sbra{f_k}\m{s}_k}. \end{split} \label{eq:AN}}

When only the rows of $\m{Y}^o$ indexed by $\Omega\subset\lbra{1,\dots,N}$ are observed, which form the submatrix $\m{Y}_{\Omega}^o$, the following atomic norm minimization problem was introduced to recover the full data matrix $\m{Y}^o$ and its frequencies \cite{yang2016exact,li2016off}:
\equ{\begin{split}
&\min_{\m{Y}} \atomn{\m{Y}},\\
&\st \m{Y}_{\Omega} = \m{Y}_{\Omega}^o. \label{eq:ANM} \end{split}}
In particular, by \eqref{eq:ANM} we attempt to find the signal that is consistent with the observed data and has the smallest atomic norm. The problem in \eqref{eq:ANM} can then be cast as an SDP in order to solve it, in two ways. The first was proposed in \cite{yang2016exact,li2016off} (and later reproduced in \cite{steffens2018compact}) by writing $\atomn{\m{Y}}$ as:
\equ{\begin{split}\atomn{\m{Y}}
=& \min_{\m{X}, \m{t}}\frac{1}{2}\tr\sbra{\m{X}} +\frac{1}{2}t_0, \\
&\st \begin{bmatrix} \m{X} & \m{Y}^H \\ \m{Y} & \m{T} \end{bmatrix} \geq \m{0}. \label{eq:SDP_AN} \end{split}}
It follows that \eqref{eq:ANM} can be equivalently written as:
\equ{\begin{split}
&\min_{\m{X}, \m{t}, \m{Y}}\frac{1}{2}\tr\sbra{\m{X}} +\frac{1}{2}t_0, \\
&\st \begin{bmatrix} \m{X} & \m{Y}^H \\ \m{Y} & \m{T} \end{bmatrix} \geq \m{0} \text{ and } \m{Y}_{\Omega} = \m{Y}_{\Omega}^o. \label{eq:SDP_ANM} \end{split}}
In \eqref{eq:SDP_AN} and \eqref{eq:SDP_ANM}, $\m{T}$ is an $N\times N$ Hermitian Toeplitz matrix and is defined as $T_{ml} = t_{l-m}$, $1\leq m\leq l\leq N$, and $\m{t}=\mbra{t_j}_{j=0}^{N-1}$.

Alternatively, an SDP can also be provided for the following dual of \eqref{eq:ANM}:
\equ{\begin{split}
&\max_{\m{V}} \inp{\m{V}_{\Omega}, \m{Y}_{\Omega}^o}_{\bR},\\
&\st \datomn{\m{V}}\leq 1 \text{ and } \m{V}_{\Omega^c} = \m{0}, \label{eq:dual} \end{split}}
where $\inp{\m{A}, \m{B}}_{\bR} = \Re\tr\sbra{\m{B}^H\m{A}}$ is the inner product of matrices $\m{A}$ and $\m{B}$.
The dual atomic norm $\datomn{\m{V}}$ is defined as:
\equ{\begin{split}\datomn{\m{V}}
&= \sup_{\m{a}\in\text{conv}\sbra{\cA}} \inp{\m{V},\m{a}}_{\bR} \\
&= \sup_{\m{a}\in\cA} \inp{\m{V},\m{a}}_{\bR} \\
&= \sup_f\twon{\m{a}^H\sbra{f}\m{V}}. \end{split} \notag}
It follows that the constraint $\datomn{\m{V}}\leq 1$ is equivalent to \equ{\twon{\m{a}^H\sbra{f}\m{V}}\leq 1 \text{ for all } f\in\bT. \notag}
Using theory of positive trigonometric polynomials \cite{dumitrescu2007positive}, the above constraint can be cast as a linear matrix inequality (LMI) so that the dual SDP is given by:
\equ{\begin{split}
&\max_{\m{V}} \inp{\m{V}_{\Omega}, \m{Y}_{\Omega}^o}_{\bR},\\
&\st \left\{ \begin{array}{l} \begin{bmatrix} \m{I} & \m{V}^H \\ \m{V} & \m{H} \end{bmatrix}\geq \m{0},\\ \tr\sbra{\m{H}} = 1, \\ \sum_{n=1}^{N-j} H_{n,n+j}=0, \; j=1,\dots,N-1,\\ \m{V}_{\Omega^c} = \m{0}. \end{array} \right. \label{eq:SDP_dual} \end{split}}

Both \eqref{eq:SDP_ANM} and \eqref{eq:SDP_dual} can be solved using off-the-shelf SDP solvers. In fact, when one is solved using a primal-dual algorithm, the solution to the other is given simultaneously. While the solved signal $\m{Y}$ is given by the primal solution, the frequencies in $\m{Y}$ can be retrieved from either the primal or the dual solution. In particular, given the primal solution $\m{T}$, the frequencies can be obtained from its so-called Vandermonde decomposition \cite{stoica2005spectral} given as:
\equ{\m{T} = \sum_{j=1}^{r} c_j\m{a}\sbra{f_j}\m{a}^H\sbra{f_j}, \label{eq:VD}}
where $r=\rank\sbra{\m{T}}$ and $c_j>0$.
This decomposition is unique in the typical case of $\rank\sbra{\m{T}}<N$ and can be computed using subspace-based methods such as ESPRIT \cite{roy1989esprit}. The atomic decomposition of $\m{Y}$ can then be obtained by solving for the coefficients $\lbra{\m{s}_j}_{j=1}^r$ from \eqref{eq:model2} using a least squares method. Interestingly, it holds that $c_j=\twon{\m{s}_j}$. This means that, besides the frequencies, their magnitudes are also given by the Vandermonde decomposition of $\m{T}$.

Given the solution $\m{V}$ in the dual, $Q(f) = \m{a}^H(f)\m{V}$ is referred to as the (vector) dual polynomial. The frequencies in $\m{Y}$ can be identified by those values $f_j$ satisfying
\equ{\twon{Q(f_j)} = \twon{\m{a}^H(f_j)\m{V}} = 1. \notag}
Subsequently, the coefficients in the atomic decomposition can be solved for in the same manner.

The dimensionality of both the primal SDP in \eqref{eq:SDP_ANM} and the dual SDP in \eqref{eq:SDP_dual} increases as the number of channels $L$ increases. To reduce the computational workload when $L$ is large, a dimensionality reduction technique was introduced in \cite{yang2016enhancing} that, in the case of $L>\rank\sbra{\m{Y}_{\Omega}^o}$, reduces the number of channels from $L$ to $\rank\sbra{\m{Y}_{\Omega}^o}$ and at the same time produces the same $\m{T}$, from which both the frequencies and their magnitudes are obtained. In particular, for any $\widetilde{\m{Y}_{\Omega}^o}$ satisfying  $\widetilde{\m{Y}_{\Omega}^o}\widetilde{\m{Y}_{\Omega}^o}^H = \m{Y}_{\Omega}^o\m{Y}_{\Omega}^{oH}$, whose number of columns may be as small as $\rank\sbra{\m{Y}_{\Omega}^o}$, the solution to $\m{T}$ remains unchanged if, in \eqref{eq:SDP_ANM}, we replace $\m{Y}_{\Omega}^o$ with $\widetilde{\m{Y}_{\Omega}^o}$ and properly change the dimensions of $\m{X}$ and $\m{Y}$. This means that the same $\m{T}$ can be obtained by solving the following SDP:
\equ{\begin{split}
&\min_{\m{X}, \m{t}, \m{Y}}\frac{1}{2}\tr\sbra{\m{X}} +\frac{1}{2}t_0,\\
&\st \begin{bmatrix} \m{X} & \m{Y}^H \\ \m{Y} & \m{T} \end{bmatrix} \geq \m{0} \text{ and } \m{Y}_{\Omega} = \widetilde{\m{Y}_{\Omega}^o}. \label{eq:SDP_DR} \end{split}}
Similar techniques were also reported in \cite{steffens2018compact} and \cite{haghighatshoar2017massive}.

Suppose that the source powers are approximately constant across different channels. Then the magnitude of $\widetilde{\m{Y}_{\Omega}^o}$ increases and may even become unbounded as $L$ increases. To render \eqref{eq:SDP_DR} solvable for large $L$, we consider a minor modification by shrinking $\widetilde{\m{Y}_{\Omega}^o}$ by $\sqrt{L}$ so that $\widetilde{\m{Y}_{\Omega}^o}\widetilde{\m{Y}_{\Omega}^o}^H = \frac{1}{L} \m{Y}_{\Omega}^o\m{Y}_{\Omega}^{oH}$ (note that $\frac{1}{L} \m{Y}_{\Omega}^o\m{Y}_{\Omega}^{oH}$ is the sample covariance matrix). Consequently, the solution to $\m{T}$ shrinks by $\sqrt{L}$, resulting in the same frequencies and scaled magnitudes. As $L$ approaches infinity, the sample covariance matrix approaches the data covariance matrix by the law of large numbers and therefore, given the data covariance matrix $\m{R}$, this modification enables us to deal with the case of $L\rightarrow \infty$ by choosing $\widetilde{\m{Y}_{\Omega}^o}$ satisfying $\widetilde{\m{Y}_{\Omega}^o}\widetilde{\m{Y}_{\Omega}^o}^H = \m{R}$. This result will be used in Section \ref{sec:simulation} to study the numerical performance of atomic norm minimization as $L\rightarrow \infty$.

Finally, if we assume that the frequencies lie on a given fixed grid, denoted by $\lbra{\tilde{f}_g}_{g=1}^G \subset \bT$, where $G$ is the grid size, then the atomic norm minimization in \eqref{eq:ANM} and \eqref{eq:SDP_ANM} can be simplified to the following $\ell_{2,1}$ norm minimization problem \cite{malioutov2005sparse}:
\equ{\begin{split}
&\min_{\lbra{\tilde{\m{s}}_g}} \sum_{g=1}^G \twon{\tilde{\m{s}}_g},\\
&\st \sum_{g=1}^G \m{a}_{\Omega}\sbra{\tilde{f}_g} \tilde{\m{s}}_g = \m{Y}_{\Omega}^o. \label{eq:l21} \end{split}}
In \eqref{eq:l21}, $\m{a}_{\Omega}(\cdot)$ is a subvector of $\m{a}(\cdot)$, and the frequency support $\cT$ can be identified from the solutions to $\lbra{\tilde{\m{s}}_g}$ that are nonzero. This $\ell_{2,1}$ norm minimization problem is used in Corollary \ref{cor:l21} for frequency estimation.

\section{Numerical Simulations} \label{sec:simulation}
We now present numerical simulations demonstrating the decrease and its rate of the required sample size per channel with the increase of the number of channels. Numerical results provided in previous publications \cite{yang2016exact,li2016off,fernandez2016super, yang2016enhancing} show the advantages of taking more channel signals in either improving the probability of successful frequency estimation or enhancing the resolution.

\begin{figure}
\centering
  \includegraphics[width=3in]{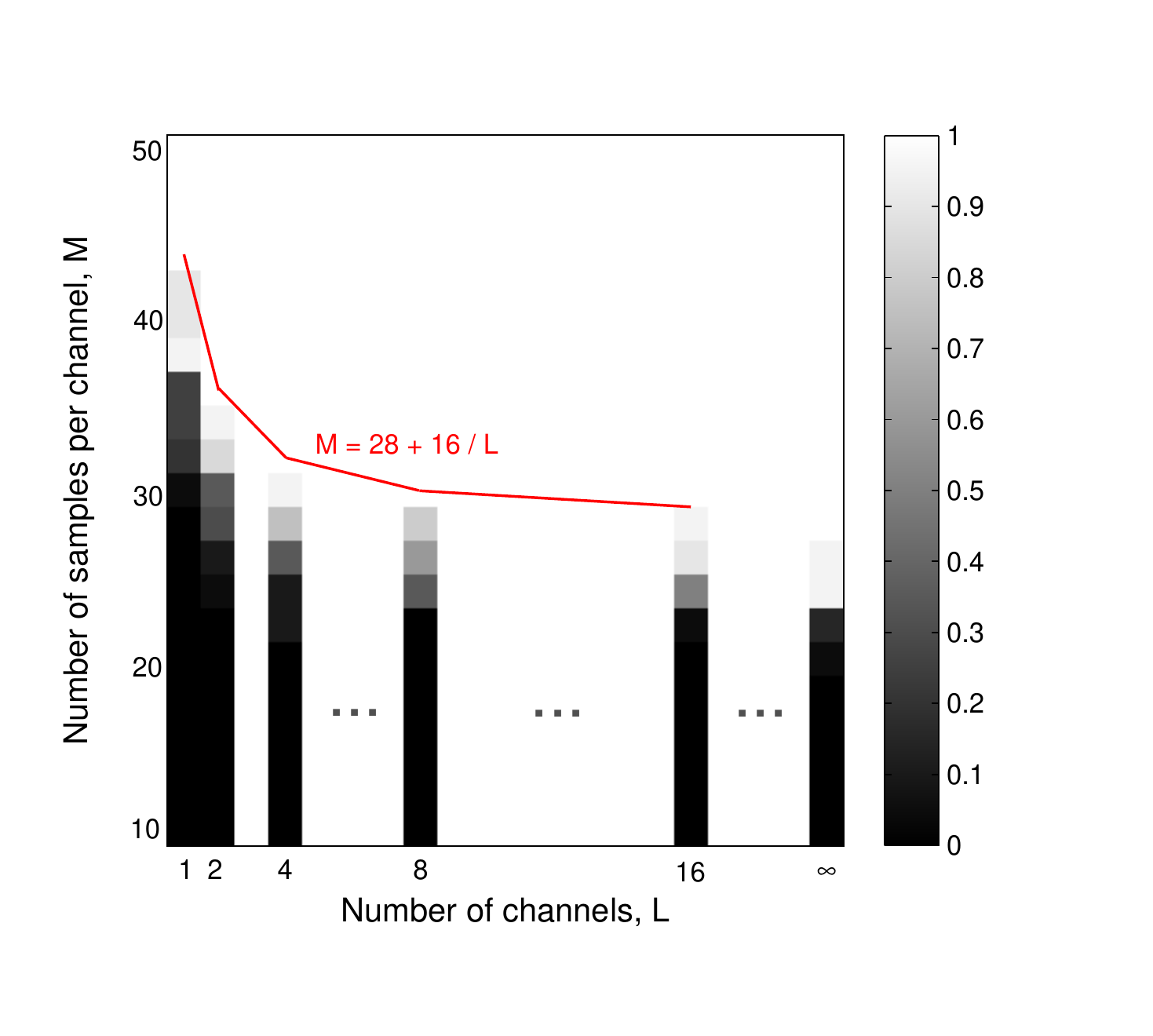}
\centering
\caption{Success rates of atomic norm minimization for multichannel frequency estimation, with $N = 128$, $K=10$, $L=1,2,4,8,16,\infty$ and $M=10,12,\dots,50$, under the minimum separation condition in \eqref{eq:separation}. White means complete success and black means complete failure. The red curve plots $M = 28+16/L$. The decreasing behavior and rate of the sample size with respect to the number of channels match those predicted in Theorem \ref{thm:noiseless}.
} \label{Fig:simulation}
\end{figure}

In our simulations, we let $N=128$, $K=10$, $L\in\lbra{1, 2, 4, 8, 16, \infty}$, and $M\in\lbra{10,12,\dots,50}$. Note that the case of $L=\infty$ can be dealt with by using the dimensionality reduction technique presented in Section \ref{sec:ANM}. For each pair $\sbra{L,M}$, 20 Monte Carlo runs are carried out. In each run, the frequencies $\lbra{f_k}$ are randomly generated and satisfy the minimum separation condition \eqref{eq:separation}, and the amplitudes $\lbra{s_{kl}}$ are independently generated from a standard complex Gaussian distribution. After the full data matrix $\m{Y}^o$ is computed according to \eqref{eq:model}, $M$ rows of $\m{Y}^o$ are randomly selected and fed into atomic norm minimization for frequency estimation implemented using CVX \cite{grant2008cvx}. The frequencies are said to be successfully estimated if the root mean squared error, computed as
\equ{\sqrt{\frac{\sum_{k=1}^K\abs{f_k - \hat f_k}^2}{K}},\notag}
is less than $10^{-4}$, where $\hat f_k$ denotes the estimate of $f_k$. We also calculate the rate of successful estimation over the 20 runs. Our simulation results are presented in Fig.~\ref{Fig:simulation}. It can be seen that the required sample size per channel for exact frequency estimation decreases with the number of channels. Evidently, the decrease rate matches well with that predicted in Theorem \ref{thm:noiseless} even in such a low-dimensional problem setup.

\section{Proof of Theorem \ref{thm:noiseless}} \label{sec:proof}
We now prove Theorem \ref{thm:noiseless}. The proof of Corollary \ref{cor:l21} is similar and is given in Appendix \ref{append:corl21}. Our proof is inspired by the analysis for multichannel frequency estimation in \cite{yang2016exact}, the average case analysis for multichannel sparse recovery in \cite{eldar2010average}, and several earlier publications on compressed sensing and super-resolution \cite{candes2006robust, candes2013towards,tang2012compressed, gribonval2008atoms}. We mainly follow the steps of \cite{yang2016exact}, with non-trivial modifications. To make the proof concise and self-contained, we first revisit the result in \cite{yang2016exact} and then highlight the key differences in our problem.

\subsection{Revisiting Previous Analysis}
The following theorem is proved in \cite{yang2016exact}.

\begin{thm} Suppose we observe the $N\times L$ data matrix $\m{Y}^o=\sum_{k=1}^K c_k \m{a}\sbra{f_k}\m{\phi}_k$
on rows indexed by $\Omega\subset\lbra{1,\dots,N}$, where $\Omega$ is of size $M$, selected uniformly at random, and given. Assume that the vector phases $\lbra{\m{\phi}_k}$ are independent and centered random variables and that the frequencies $\lbra{f_k}$ satisfy the same separation condition as in Theorem \ref{thm:noiseless}. Then, there exists a numerical constant $C$ such that
\equ{M\geq C\max\lbra{\log^2\frac{\sqrt{L}N}{\delta}, K\log\frac{K}{\delta}\log\frac{\sqrt{L}N}{\delta}} \label{eq:AN_bound1}}
is sufficient to guarantee that, with probability at least $1-\delta$, the full data $\m{Y}^o$ and its frequency support $\cT$ can be uniquely produced by solving \eqref{eq:ANM} [or equivalently \eqref{eq:SDP_ANM}]. \label{thm:noiseless1}
\end{thm}

According to \cite{yang2016exact}, the optimality of a solution to \eqref{eq:ANM} can be validated using a dual certificate that is provided in the following proposition.

\begin{prop} The matrix $\m{Y}^o=\sum_{k=1}^K c_k\m{a}\sbra{f_k}\m{\phi}_k$ is the unique optimizer of \eqref{eq:ANM} if $\lbra{\m{a}_{\Omega}\sbra{f_k}}_{f_k\in\cT}$ are linearly independent and if there exists a vector-valued dual polynomial $\overline Q: \bT\rightarrow \bC^{1\times L}$
\equ{\overline Q(f)=\m{a}^H(f)\m{V} \label{eq:dualpoly1_inc}}
satisfying
\lentwo{\equa{\overline  Q\sbra{f_k}
&=& \m{\phi}_k, \quad f_k\in\cT, \label{eq:cons1}\\ \twon{\overline Q\sbra{f}}
&<& 1, \quad f\in\bT\backslash\cT, \label{eq:cons2}\\ \m{V}_j
&=& \m{0}, \quad j\in \Omega^c, \label{eq:cons3}
}}where $\m{V}$ is an $N\times L$ matrix and $\m{V}_j$ denotes its $j$th row. Moreover, $\m{Y}^o=\sum_{k=1}^K c_k\m{a}\sbra{f_k}\m{\phi}_k$ is the unique atomic decomposition achieving the atomic norm with $\atomn{\m{Y}^o}=\sum_{k=1}^K c_k$. \label{prop:dualpoly}
\end{prop}

Applying Proposition \ref{prop:dualpoly}, Theorem \ref{thm:noiseless1} is proved in \cite{yang2016exact} by showing the existence of an appropriate dual polynomial $\overline Q(f)$ under the assumptions of Theorem \ref{thm:noiseless1} (note that in this process the condition of linear independence of $\lbra{\m{a}_{\Omega}\sbra{f_k}}_{f_k\in\cT}$ is also satisfied). To explicitly construct the dual polynomial $\overline Q(f)$, it is shown in \cite{yang2016exact} that we may consider the symmetric case where the rows of $\m{Y}^o$ are indexed by $\cJ \coloneqq \lbra{-2n,\dots,2n}$ instead of $\lbra{1,\dots,N}$, where $N = 4n + n_0$ with $n_0=1,2,3,4$. Moreover, we may equivalently consider the Bernoulli observation model following from \cite{candes2006robust,tang2012compressed}, in which each row of $\m{Y}^o$ is observed independently with probability $p=\frac{M}{4n}$, rather than the uniform observation model as in Theorems \ref{thm:noiseless} and \ref{thm:noiseless1}. The remainder of the proof consists of two steps.
\begin{enumerate}
\item Consider the full data case where $\Omega^c$ is empty and construct $\overline Q(f)$, referred to as $Q(f)$, satisfying \eqref{eq:cons1} and
    \lentwo{\equa{
    &&\twon{Q(f)}^2\leq 1-c_1n^2 (f-f_k)^2,\notag \\
    &&\qquad f\in\cN_k\coloneqq \sbra{f_k-\frac{0.16}{n},\; f_k+\frac{0.16}{n}}, \label{eq:Qnear}\\
    &&\twon{Q(f)}\leq 1-c_2, \quad f\in \cF \coloneqq \bT\setminus \bigcup_{k=1}^K \cN_k, \label{eq:Qfar}\\
    && \frac{\text{d}^2\twon{Q(f)}^2}{\text{d}f^2}\leq -2c_1n^2,\quad f\in\cN_k, \label{eq:Qnear1}
    }}where $c_1$ and $c_2$ are positive constants. Note that \eqref{eq:Qnear} and \eqref{eq:Qfar} form a stronger condition than \eqref{eq:cons2}.
\item In the compressive data case of interest, construct a random polynomial $\overline Q(f)$ satisfying \eqref{eq:cons1} and \eqref{eq:cons3} with respect to the Bernoulli sampling scheme. Show that, if $M$ satisfies \eqref{eq:AN_bound1}, then $\overline Q(f)$ (and its first and second derivatives) is close to $Q(f)$ (and its first and second derivatives) on the whole unit circle $\bT$ with high probability under the assumptions of Theorem \ref{thm:noiseless1}, and $\overline Q(f)$ also satisfies \eqref{eq:Qnear} and \eqref{eq:Qfar}. As a result, $\overline Q(f)$ is a polynomial as required in Proposition \ref{prop:dualpoly}.
\end{enumerate}

In the ensuing subsections, we prove Theorem \ref{thm:noiseless} following the aforementioned steps of the proof of Theorem \ref{thm:noiseless1}. The difference is in the second step above, namely, showing that $\overline Q(f)-Q(f)$ and its derivatives are arbitrarily small on the unit circle $\bT$ with high probability under the assumptions of Theorem \ref{thm:noiseless}, with $M$ in \eqref{eq:AN_bound} being a decreasing function of the number of channels $L$. In particular, we show in Subsection \ref{sec:Qf} how the dual certificate $\overline Q(f)$ is constructed, summarize in Subsection \ref{sec:lemma} some useful lemmas shown in \cite{candes2013towards,tang2012compressed,yang2016exact}, analyze in Subsection \ref{sec:diff} the difference $\overline Q(f)-Q(f)$, and complete the proof in Subsection \ref{sec:complete}.

\subsection{Construction of $\overline Q(f)$} \label{sec:Qf}
To construct $\overline Q(f)$, we start with the squared Fej\'{e}r kernel
\equ{\cK\sbra{f}=\mbra{\frac{\sin(\pi (n+1)f)}{(n+1)\sin\sbra{\pi f}}}^4=\sum_{j=-2n}^{2n}g_n\sbra{j}e^{-i2\pi jf} \label{eq:K}}
with coefficients
\equ{\begin{split}
&g_n\sbra{j}\\
&=\frac{1}{n+1}\sum_{k=\max\sbra{j-n-1,-n-1}}^{\min\sbra{j+n+1,n+1}} \sbra{1-\frac{\abs{k}}{n+1}}\sbra{1-\frac{\abs{j-k}}{n+1}} \end{split} \notag}
obeying $0<g_n\sbra{j}\leq1$, $j=-2n,\dots,2n$. This kernel equals unity at the origin and decays rapidly away from it. Let $\lbra{\delta_j}_{j\in\cJ}$ be i.i.d.~Bernoulli random variables with
\equ{\bP\sbra{\delta_j=1} = p = \frac{M}{4n}.\notag }
It follows that
\equ{\Omega = \lbra{j\in\cJ:\;\delta_j = 1} \notag}
with $\bE \abs{\Omega} \approx M$.
This allows us to write a compressive-data analog of the squared Fej\'{e}r kernel as
\equ{\overline \cK\sbra{f}= \sum_{j\in\Omega} g_n\sbra{j}e^{-i2\pi jf} = \sum_{j=-2n}^{2n} \delta_j g_n\sbra{j}e^{-i2\pi jf} \label{eq:Kb}}
and define the vector-valued polynomial $\overline Q(f)$ as
\equ{\overline Q\sbra{f}=\sum_{f_k\in\cT}\m{\alpha}_k \overline\cK\sbra{f-f_k} + \sum_{f_k\in\cT}\m{\beta}_k \overline\cK^{\sbra{1}}\sbra{f-f_k}, \label{eq:randpoly}}
where $\m{\alpha}_k$ and $\m{\beta}_k$ are $1\times L$ vector coefficients to specify and the superscript $l$ denotes the $l$th derivative.

It is evident that $\overline Q\sbra{f}$ in \eqref{eq:randpoly} satisfies the support condition in \eqref{eq:cons3}. According to \cite{tang2012compressed}, $\overline\cK$ and its derivatives are concentrated around their expectations, $p\cK$ and its derivatives, if $M$ is large enough. As a result, $\overline Q\sbra{f}$ is expected to peak near $f_k\in\cT$ if $\overline Q\sbra{f}$ is dominated by the first term and the coefficients $\m{\alpha}_k$ and $\m{\beta}_k$ are appropriately chosen. To make $\overline Q\sbra{f}$ satisfy \eqref{eq:cons1}, we impose for any $f_j\in\cT$,
\equ{\sum_{f_k\in\cT}\m{\alpha}_k \overline\cK\sbra{f_j-f_k} + \sum_{f_k\in\cT} \m{\beta}_k\overline\cK^{\sbra{1}}\sbra{f_j-f_k}=\m{\phi}_j. \label{eq:Qfj}}
To satisfy \eqref{eq:cons2}, a necessary condition is that the derivative of $\twon{\overline Q\sbra{f}}^2$ vanishes at $f_j\in\cT$, leading to
\equ{\begin{split}\frac{\text{d}\twon{\overline Q\sbra{f}}^2}{\text{d}f}\Bigg|_{f=f_j}
&= 2\Re\lbra{\overline Q^{\sbra{1}}\sbra{f_j} \overline Q^{H}\sbra{f_j}}\\
&= 2\Re\lbra{ \overline Q^{\sbra{1}}\sbra{f_j}\m{\phi}_j^H}\\
&= 0. \label{eq:Qdfj1} \end{split}}
Consequently, one feasible choice is to let, for any $f_j\in\cT$,
\equ{\begin{split}
&\overline Q^{\sbra{1}}\sbra{f_j} \\
&= \sum_{f_k\in\cT}\m{\alpha}_k \overline\cK^{\sbra{1}}\sbra{f_j-f_k} + \sum_{f_k\in\cT}\m{\beta}_k\overline\cK^{\sbra{2}}\sbra{f_j-f_k}\\
&=\m{0}. \label{eq:Qdfj} \end{split}}

The conditions in \eqref{eq:Qfj} and \eqref{eq:Qdfj} consist of $2KL$ equations and therefore may determine the coefficients $\lbra{\m{\alpha}_k}$ and $\lbra{\m{\beta}_k}$ (and hence $\overline Q\sbra{f}$). Define the $K\times K$ matrices $\overline {\m{D}}_l$, $l=0,1,2$ such that $\mbra{\overline D_l}_{jk} = \overline \cK^{\sbra{l}}\sbra{f_j-f_k}$ (note that the zeroth derivative is itself). Then, \eqref{eq:Qfj} and \eqref{eq:Qdfj} lead to the following system of linear equations:
\equ{\overline{\m{D}}\begin{bmatrix}\m{\alpha}\\ c_0\m{\beta}\end{bmatrix} = \begin{bmatrix}\overline{\m{D}}_0 & c_0^{-1}\overline{\m{D}}_1 \\ -c_0^{-1}\overline{\m{D}}_1 & -c_0^{-2}\overline{\m{D}}_2\end{bmatrix} \begin{bmatrix}\m{\alpha}\\ c_0\m{\beta}\end{bmatrix} = \begin{bmatrix}\m{\Phi}\\ \m{0}\end{bmatrix}, \label{eq:linsys}}
where $\m{\Phi}$ is the $K\times L$ matrix formed by stacking $\lbra{\m{\phi}_j}$ together, with $\m{\alpha}$ and $\m{\beta}$ similarly defined. The constant $c_0=\sqrt{\abs{\cK^{\sbra{2}}\sbra{0}}}=\sqrt{\frac{4\pi^2n(n+2)}{3}}$ is introduced so that the coefficient matrix $\overline{\m{D}}$ is symmetric and well-conditioned \cite{tang2012compressed}. Note that all compressive-data (random) quantities defined above such as $\overline{Q}$ and $\overline{\m{D}}$ have full-data (deterministic) analogs, denoted by $Q$ and $\m{D}$ by removing the overlines and obtained by replacing $\overline \cK$ in their expressions with $\cK$ [refer $Q$ to \eqref{eq:Qnear} and \eqref{eq:Qfar}]. The remaining task is to show that, if $M$ satisfies \eqref{eq:AN_bound}, then $\overline{\m{D}}$ is invertible and $\overline Q\sbra{f}$ can be uniquely determined under the assumptions of Theorem \ref{thm:noiseless}. In addition, show that $\overline Q\sbra{f}$ satisfies \eqref{eq:cons2}, which together with Proposition \ref{prop:dualpoly} completes the proof.

\subsection{Useful Lemmas} \label{sec:lemma}
To complete the proof we rely on some useful results shown in \cite{candes2013towards,tang2012compressed, yang2016exact} that are summarized below.

First consider the invertibility of $\overline{\m{D}}$. For $\tau\in\left(0,\frac{1}{4}\right]$, define the event
\equ{\cE_{1,\tau}\coloneqq\lbra{\twon{p^{-1}\overline{\m{D}}-\m{D}}\leq\tau}. \notag}
Let $\delta$, $\epsilon$ be small positive numbers and $C$ a constant, which are independent of the parameters $K$, $M$, $N$ and $L$, and may vary from instance to instance. Then we have the following lemma which, in addition to the invertibility of $\overline{\m{D}}$, also guarantees linear independence of $\lbra{\m{a}_{\Omega}\sbra{f_k}}_{f_k\in\cT}$ as required in Proposition \ref{prop:dualpoly}.
\begin{lem}[\cite{candes2013towards,tang2012compressed, yang2016exact}] Assume $\Delta_{\cT}\geq \frac{1}{n}$ and $n\geq 64$ and let $\tau\in\left(0,\frac{1}{4}\right]$. Then, $\m{D}$ is invertible, $\overline{\m{D}}$ is invertible on $\cE_{1,\tau}$, $\lbra{\m{a}_{\Omega}\sbra{f_k}}_{f_k\in\cT}$ are linearly independent on $\cE_{1,\tau}$, and $\bP\sbra{\cE_{1,\tau}}\geq 1-\delta$ if
\equ{M\geq\frac{50}{\tau^2}K\log\frac{2K}{\delta}. \notag}  \label{lem:invertibility}
\end{lem}

On $\cE_{1,\tau}$, we introduce the partitions $\overline{\m{D}}^{-1}=\begin{bmatrix}\overline{\m{L}} & \overline{\m{R}}\end{bmatrix}$ and $\m{D}^{-1}=\begin{bmatrix}\m{L} & \m{R}\end{bmatrix}$, where $\overline{\m{L}}$, $\overline{\m{R}}$, $\m{L}$ and $\m{R}$ are all $2K\times K$ matrices. For $l=0,1,2$, let
\equ{\overline{\m{v}}_l\sbra{f}=c_0^{-l}\begin{bmatrix} \overline\cK^{\sbra{l}H}\sbra{f-f_1} \\ \vdots \\ \overline\cK^{\sbra{l}H}\sbra{f-f_K} \\ c_0^{-1}\overline\cK^{\sbra{l+1}H}\sbra{f-f_1} \\ \vdots \\c_0^{-1} \overline\cK^{\sbra{l+1}H}\sbra{f-f_K} \end{bmatrix} \label{eq:vl}}
and similarly define its deterministic analog $\m{v}_l\sbra{f}$, where $f_j\in\cT$ and the superscript $H$ denotes the complex conjugate for a scalar. It follows that on the event $\cE_{1,\tau}$,
{\lentwo\equa{ \begin{bmatrix}\m{\alpha}\\ c_0\m{\beta}\end{bmatrix}
&=&\overline{\m{D}}^{-1}\begin{bmatrix}\m{\Phi}\\ \m{0}\end{bmatrix} =\overline{\m{L}}\m{\Phi}, \\ c_0^{-l}\overline Q^{\sbra{l}}\sbra{f}
&=& \sum_{f_k\in\cT} \m{\alpha}_k c_0^{-l} \overline\cK^{\sbra{l}}\sbra{f-f_k} \notag \\
&& + \sum_{f_k\in\cT} c_0\m{\beta}_k\cdot c_0^{-\sbra{l+1}}\overline\cK^{\sbra{l+1}}\sbra{f-f_k} \notag\\
&=& \overline{\m{v}}_l^H\sbra{f}\overline{\m{L}}\m{\Phi} \notag\\
&\coloneqq& \inp{\m{\Phi}, \overline{\m{L}}^H \overline{\m{v}}_l\sbra{f}}, \label{eq:Qdc}}
where the notation of inner-product is abused since the result of the product is a vector rather than a scalar.

We write $\overline{\m{L}}^H \overline{\m{v}}_l\sbra{f}$ in \eqref{eq:Qdc} in three parts:
\equ{\begin{split}\overline{\m{L}}^H \overline{\m{v}}_l\sbra{f}
&= \m{L}^H\m{v}_l\sbra{f} +\overline{\m{L}}^H\sbra{\overline{\m{v}}_l\sbra{f}-p\m{v}_l\sbra{f} } \\
&\quad + \sbra{\overline{\m{L}}-p^{-1}\m{L}}^Hp\m{v}_l\sbra{f},
\end{split} \notag}
which results in the following decomposition of $c_0^{-l}\overline Q^{\sbra{l}}\sbra{f}$:
\equ{\begin{split}c_0^{-l}\overline Q^{\sbra{l}}\sbra{f}
&= \inp{\m{\Phi}, \overline{\m{L}}^H \overline{\m{v}}_l\sbra{f}} \\
&= \inp{\m{\Phi}, \m{L}^H\m{v}_l\sbra{f}} \\
&\quad + \inp{\m{\Phi}, \overline{\m{L}}^H\sbra{\overline{\m{v}}_l\sbra{f}-p\m{v}_l\sbra{f} }}\\
&\quad + \inp{\m{\Phi}, \sbra{\overline{\m{L}}-p^{-1}\m{L}}^Hp\m{v}_l\sbra{f}}\\
&= c_0^{-l}Q^{\sbra{l}}\sbra{f} +I_1^l\sbra{f} + I_2^l\sbra{f},
\end{split} \label{eq:decomposition}}
where we have defined
\lentwo{\equa{ I_1^l\sbra{f}
&\coloneqq& \inp{\m{\Phi}, \overline{\m{L}}^H\sbra{\overline{\m{v}}_l\sbra{f}-p\m{v}_l\sbra{f} }}, \notag \\ I_2^l\sbra{f}
&\coloneqq& \inp{\m{\Phi}, \sbra{\overline{\m{L}}-p^{-1}\m{L}}^Hp\m{v}_l\sbra{f}}.\notag
}}

As a result of \eqref{eq:decomposition}, a connection between $\overline Q\sbra{f}$ and $Q\sbra{f}$ is established. Our goal is to show that the random perturbations $I_1^l\sbra{f} + I_2^l\sbra{f}$, $l=0,1,2$ can be arbitrarily small when $M$ is sufficiently large, meaning that $c_0^{-l}\overline Q^{\sbra{l}}\sbra{f}$ is concentrated around $c_0^{-l} Q^{\sbra{l}}\sbra{f}$. To this end, we need the following results shown in \cite{tang2012compressed}.

\begin{lem}[\cite{tang2012compressed}] Assume $\Delta_{\cT}\geq \frac{1}{n}$, let $\tau\in\left(0,\frac{1}{4}\right]$, and consider a finite set $\bT_{\text{grid}}=\lbra{f_d}\subset\bT$. Then, we have
\equ{\begin{split}
& \bP\left[\sup_{f_d\in\bT_{\text{grid}}} \twon{\overline{\m{L}}^H\sbra{\overline{\m{v}}_l\sbra{f_d} -p\m{v}_l\sbra{f_d}}}\right.\\
&\quad  \left.\geq 4\sbra{2^{2l+3}\sqrt{\frac{K}{M}}+\frac{n}{M}a\overline{\sigma}_l},\; l=0,1,2 \right] \\
& \leq 64\abs{\bT_{\text{grid}}}e^{-\gamma a^2}+\bP\sbra{\cE_{1,\tau}^c}
\end{split} \notag}
for some constant $\gamma>0$, where $\overline\sigma_l^2= 2^{4l+1}\frac{M}{n^2}\max\lbra{1,\frac{2^4K}{\sqrt{M}}}$ and
\equ{0<a\leq \left\{\begin{array}{ll} \sqrt{2}M^{\frac{1}{4}}, & \text{if } \frac{2^4K}{\sqrt{M}}\geq1, \\ \frac{\sqrt{2}}{4}\sqrt{\frac{M}{K}}, & \text{otherwise.}\end{array}\right.\notag} \label{lem:I1bound}
\end{lem}

\begin{lem}[\cite{tang2012compressed}] Assume $\Delta_{\cT}\geq \frac{1}{n}$. On the event $\cE_{1,\tau}$, we have
\equ{\twon{\sbra{\overline{\m{L}}-p^{-1}\m{L} }^Hp\m{v}_l\sbra{f} } \leq C\tau \notag}
for some constant $C>0$. \label{lem:I2bound}
\end{lem}

\subsection{Analysis of $\overline Q(f)-Q(f)$ and Its Derivatives} \label{sec:diff}
In this subsection we show that the quantities $c_0^{-l}\overline Q^{\sbra{l}}\sbra{f} - c_0^{-l} Q^{\sbra{l}}\sbra{f}$, $l=0,1,2$ can be arbitrarily small on the unit circle. We first consider a set of finite grid points $\bT_{\text{grid}}\subset \bT$ and define the event
\equ{\cE_2\coloneqq\lbra{\sup_{f_d\in\bT_{\text{grid}}} c_0^{-l} \twon{\overline Q^{\sbra{l}}\sbra{f_d}-Q^{\sbra{l}}\sbra{f_d}} \leq \epsilon,\; l=0,1,2}.\notag}
The following result states that $\cE_2$ occurs with high probability if $M$ is sufficiently large.

\begin{prop} Suppose $\bT_{\text{grid}}\subset \bT$ is a finite set of grid points. Under the assumptions of Theorem \ref{thm:noiseless}, there exists a numerical constant $C$ such that if
\equ{\begin{split}M\geq C\frac{1}{\epsilon^2}\max&\lbra{\log^2\frac{\abs{\bT_{\text{grid}}}}{\delta}, \right.\\
&\quad \left.K\log\frac{K}{\delta} \sbra{1+\frac{1}{L}\log\frac{\abs{\bT_{\text{grid}}}}{\delta}}}, \label{eq:Mboundgrid} \end{split}}
then
\equ{\bP\sbra{\cE_2}\geq 1-\delta. \notag} \label{prop:boundgrid}
\end{prop}

To prove Proposition \ref{prop:boundgrid}, we need to show that both $I_1^l\sbra{f}$ and $I_2^l\sbra{f}$ are small on $\bT_{\text{grid}}$. To make the lower bound on $M$ decrease as $L$ increases, inspired by \cite{eldar2010average}, we use the following lemma that generalizes the Bernstein inequality for Steinhaus sequences in \cite[Proposition 16]{tropp2008conditioning} to higher dimensions.

\begin{lem}[\cite{eldar2010average}] Let $\m{0}\neq \m{w}\in\bC^K$ and $\lbra{\m{\phi}_k}_{k=1}^K$ be a series of independent random vectors that are uniformly distributed on the complex sphere $\bS^{2L-1}$. Then, for all $t> \twon{\m{w}}$,
\equ{\bP\sbra{\twon{\sum_{k=1}^K w_k\m{\phi}_k}\geq t } \leq e^{-L\sbra{\frac{t^2}{\twon{\m{w}}^2}-\log\frac{t^2}{\twon{\m{w}}^2} -1}}. \label{eq:bernprob}} \label{lem:bernstein}
\end{lem}

The following result will be used to move the dependence on $L$ from the upper bound on the probability in \eqref{eq:bernprob} to the sample size $M$.
\begin{lem} Let $y(x)\geq 1$ be the solution to the equation $y-\log y -1 = x$ for $x\geq 0$. Then, $y(x)$ is monotonically increasing in $x$ and $1+x \leq y(x) \leq 2(1+x)$. \label{lem:ymlogy}
\end{lem}
\begin{proof} See Appendix \ref{append:ymlogy}.
\end{proof}

Applying Lemmas \ref{lem:bernstein} and \ref{lem:ymlogy}, we show in the following two lemmas that both $I_1^l\sbra{f}$ and $I_2^l\sbra{f}$ can be arbitrarily small on $\bT_{\text{grid}}$ with $M$ being a decreasing function of $L$.
\begin{lem} Under the assumptions of Theorem \ref{thm:noiseless}, there exists a numerical constant $C$ such that if
\equ{\begin{split}M
\geq C\max&\left\{\frac{1}{\epsilon^2} K\sbra{1+\frac{1}{L}\log\frac{\abs{\bT_{\text{grid}}}}{\delta}}, \right. \\
&\quad \left.\frac{1}{\epsilon^2} \log^2\frac{\abs{\bT_{\text{grid}}}}{\delta}, \; K\log\frac{K}{\delta}\right\}, \end{split} \notag}
then we have
\equ{\bP\lbra{\sup_{f_d\in\bT_{\text{grid}}} \twon{I_1^l\sbra{f_d}}\leq\epsilon,\; l=0,1,2 } \geq 1-9\delta. \notag} \label{lem:I1}
\end{lem}
\begin{proof} See Appendix \ref{append:I1}.
\end{proof}

\begin{lem} Under the assumptions of Theorem \ref{thm:noiseless}, there exists a numerical constant $C$ such that if
\equ{M\geq C\frac{1}{\epsilon^2} K\log\frac{K}{\delta} \sbra{1+\frac{1}{L}\log\frac{\abs{\bT_{\text{grid}}}}{\delta}}, \notag}
then we have
\equ{\bP\lbra{\sup_{f_d\in\bT_{\text{grid}}} \twon{I_2^l\sbra{f_d}}<\epsilon,\; l=0,1,2 } \geq 1-6\delta.\notag} \label{lem:I2}
\end{lem}
\begin{proof} See Appendix \ref{append:I2}.
\end{proof}

Proposition \ref{prop:boundgrid} is a direct consequence of combining Lemmas \ref{lem:I1} and \ref{lem:I2}. We next extend Proposition \ref{prop:boundgrid} from the set of finite grid points $\bT_{\text{grid}}$ to the whole unit circle $\bT$. To this end, for any $f\in\bT$, we make the following decomposition:
\equ{\begin{split}
&\overline Q^{\sbra{l}}\sbra{f} - Q^{\sbra{l}}\sbra{f}\\
&= \sbra{\overline Q^{\sbra{l}}\sbra{f} - \overline Q^{\sbra{l}}\sbra{f_d}} + \sbra{\overline Q^{\sbra{l}}\sbra{f_d} - Q^{\sbra{l}}\sbra{f_d}}\\
&\quad + \sbra{Q^{\sbra{l}}\sbra{f_d} - Q^{\sbra{l}}\sbra{f}}. \label{eq:decomp}\end{split}}
Because \eqref{eq:decomp} holds for any $f_d\in\bT_{\text{grid}}$, the inequalities in \eqref{eq:decomp2} then follow (see the next page).

\newcounter{mytempeqncnt}
\begin{figure*}[!t]
\normalsize
\setcounter{mytempeqncnt}{\value{equation}}

\begin{equation}
\setcounter{equation}{38}
\begin{split}
&\sup_{f\in\bT}c_0^{-l}\twon{\overline Q^{\sbra{l}}\sbra{f} - Q^{\sbra{l}}\sbra{f}}\\
&\leq \sup_{f\in\bT}\inf_{f_d\in\bT_{\text{grid}}}\left\{c_0^{-l}\twon{\sbra{\overline Q^{\sbra{l}}\sbra{f} - \overline Q^{\sbra{l}}\sbra{f_d}} + \sbra{Q^{\sbra{l}}\sbra{f_d} - Q^{\sbra{l}}\sbra{f}}} + c_0^{-l}\twon{\overline Q^{\sbra{l}}\sbra{f_d} - Q^{\sbra{l}}\sbra{f_d}}\right\}\\
&\leq \sup_{f\in\bT}\inf_{f_d\in\bT_{\text{grid}}}c_0^{-l}\twon{\sbra{\overline Q^{\sbra{l}}\sbra{f} - \overline Q^{\sbra{l}}\sbra{f_d}} + \sbra{Q^{\sbra{l}}\sbra{f_d} - Q^{\sbra{l}}\sbra{f}}} + \sup_{f_d\in\bT_{\text{grid}}}c_0^{-l}\twon{\overline Q^{\sbra{l}}\sbra{f_d} - Q^{\sbra{l}}\sbra{f_d}}.
\end{split} \label{eq:decomp2}
\end{equation}
\setcounter{equation}{\value{mytempeqncnt}+1}
\hrulefill
\vspace*{4pt}
\end{figure*}

The second term in \eqref{eq:decomp2} can be arbitrarily small according to Proposition \ref{prop:boundgrid}. We next show that the first term can also be arbitrarily small. Recall that $c_0^{-l}\overline Q^{\sbra{l}}\sbra{f}
= \inp{\m{\Phi}, \overline{\m{L}}^H \overline{\m{v}}_l\sbra{f}}$, $c_0^{-l} Q^{\sbra{l}}\sbra{f}
= \inp{\m{\Phi}, \m{L}^H \m{v}_l\sbra{f}}$, and thus
\equ{\begin{split}
&c_0^{-l}\sbra{\overline{Q}^{\sbra{l}}\sbra{f_d} - \overline{Q}^{\sbra{l}}\sbra{f}} + c_0^{-l}\sbra{Q^{\sbra{l}}\sbra{f_d} - Q^{\sbra{l}}\sbra{f}}\\
&= \inp{\m{\Phi}, \overline{\m{L}}^H \sbra{\overline{\m{v}}_l\sbra{f_d} - \overline{\m{v}}_l\sbra{f}} + \m{L}^H \sbra{\m{v}_l\sbra{f_d} - \m{v}_l\sbra{f}}}. \label{eq:Qdc2} \end{split}}
The magnitudes of $\m{L}^H\sbra{\m{v}_l\sbra{f} - \m{v}_l\sbra{f_d}}$ and $\overline{\m{L}}^H\sbra{\overline{\m{v}}_l\sbra{f} - \overline{\m{v}}_l\sbra{f_d}}$ in \eqref{eq:Qdc2} are controlled in the following lemma.
\begin{lem} Assume $\Delta_{\cT}\geq \frac{1}{n}$. On the event $\cE_{1,\tau}$, we have
{\lentwo\equa{\twon{\m{L}^H\sbra{\m{v}_l\sbra{f} - \m{v}_l\sbra{f_d}}}
&\leq& Cn^2\abs{f-f_d}, \\ \twon{\overline{\m{L}}^H\sbra{\overline{\m{v}}_l\sbra{f} - \overline{\m{v}}_l\sbra{f_d}}}
&\leq& Cn^3\abs{f-f_d} \label{eq:Lvrand}}
}for $f,f_d\in\bT$ and some constant $C>0$. \label{lem:extendcontin}
\end{lem}
\begin{proof} See Appendix \ref{append:extendcontin}.
\end{proof}

Applying Lemmas \ref{lem:extendcontin}, \ref{lem:bernstein} and \ref{lem:ymlogy}, we show in the following lemma that the first term on the right hand side of \eqref{eq:decomp2} can be arbitrarily small if $\bT_{\text{grid}}$ is properly chosen, where the grid size $\abs{\bT_{\text{grid}}}$ decreases with $L$.

\begin{lem} Suppose $\bT_{\text{grid}}$ is a finite set of uniform grid points. Under the assumptions of Theorem \ref{thm:noiseless}, there exists a constant $C$ such that if $M\geq CK\log \frac{K}{\delta}$ and
\equ{\abs{\bT_{\text{grid}}} = \left\lceil C\frac{n^3}{\epsilon}\sqrt{1+\frac{1}{L}\log \frac{1}{\delta}} \right\rceil, \label{Tgrid}}
then
\equ{\begin{split}&\bP\Bigg(
\sup_{f\in\bT} \inf_{f_d\in\bT_{\text{grid}}} c_0^{-l}\left\|\sbra{\overline{Q}^{\sbra{l}}\sbra{f} - \overline{Q}^{\sbra{l}}\sbra{f_d}} \right. \\
&\phantom{Q^{\sbra{l}}\sbra{f_d}}\left.+ \sbra{Q^{\sbra{l}}\sbra{f_d} - Q^{\sbra{l}}\sbra{f}} \right\|_2 < \epsilon, \; l=0,1,2 \Bigg) \\
&\geq 1-6\delta. \end{split} \notag}\label{lem:Qfdf}
\end{lem}
\begin{proof} See Appendix \ref{append:Qfdf}.
\end{proof}

Combining Proposition \ref{prop:boundgrid} and Lemma \ref{lem:Qfdf}, we have the following proposition, where the bound on $M$ in \eqref{eq:Mboundcontinu} is obtained by modifying \eqref{eq:Mboundgrid} using \eqref{Tgrid} with the relaxation
\equ{\begin{split}\abs{\bT_{\text{grid}}}
&= \left\lceil C\frac{n^3}{\epsilon}\sqrt{1+\frac{1}{L}\log \frac{1}{\delta}} \right\rceil \\ &\leq C\frac{n^3}{\epsilon}\sqrt{2\log\frac{1}{\delta}} \\
&\leq \sqrt{2}C\frac{n^3}{\epsilon\sqrt{\delta}}. \end{split} \notag}
\begin{prop} Under the assumptions of Theorem \ref{thm:noiseless}, there exists a numerical constant $C$ such that if
\equ{M\geq C\frac{1}{\epsilon^2}\max\lbra{\log^2\frac{n}{\epsilon\delta}, K\log\frac{K}{\delta} \sbra{1+\frac{1}{L}\log\frac{n}{\epsilon\delta}} }, \label{eq:Mboundcontinu}}
then with probability $1-\delta$, we have
\equ{\sup_{f\in\bT}c_0^{-l}\twon{\overline Q^{\sbra{l}}\sbra{f} - Q^{\sbra{l}}\sbra{f}}\leq \epsilon, \quad l=0,1,2. \label{eq:QbQl}} \label{prop:appendcontinuous}
\end{prop}

\subsection{Completion of the Proof} \label{sec:complete}
We have shown by Proposition \ref{prop:appendcontinuous} that $\overline Q(f)$ (and its derivatives) can be arbitrarily close to $Q(f)$ (and its derivatives) with high probability provided that $M$ satisfies \eqref{eq:Mboundcontinu}. Following the same steps as in \cite{yang2014exact1}, we can show that $\overline Q(f)$ also satisfies \eqref{eq:Qnear}, \eqref{eq:Qfar} and \eqref{eq:Qnear1}, if $\epsilon$ is taken to be a small value. In particular, it follows from \eqref{eq:Qfar} and \eqref{eq:QbQl} that for $f\in \cF$,
\equ{\begin{split}\twon{\overline Q(f)}
&\leq \twon{Q(f)} + \twon{\overline Q(f)-Q(f)} \\
&\leq 1-c_2 + \epsilon. \label{eq:Qbarbound} \end{split}}
For $f\in\cN_k$, it follows from \eqref{eq:QbQl} and the proof of \cite[Lemma 6.8]{yang2014exact1} that
\equ{\begin{split}
&\abs{{\frac{{\text{d}^2}\twon{\overline Q\sbra{f}}^2}{\text{d}f^2} - \frac{{\text{d}^2} \twon{Q\sbra{f}}^2}{\text{d}f^2} }} \\
&\leq \sbra{8C_1\epsilon + 4\epsilon^2}c_0^2\\
&\leq \frac{8}{3}\pi^2\sbra{8C_1\epsilon + 4\epsilon^2}n^2, \end{split} \notag}
where $C_1$ is a constant. Then,
\equ{\begin{split}
&\frac{{\text{d}^2}\twon{\overline Q\sbra{f}}^2}{\text{d}f^2}\\
&\leq \frac{{\text{d}^2} \twon{Q\sbra{f}}^2}{\text{d}f^2} +
\abs{{\frac{{\text{d}^2}\twon{\overline Q\sbra{f}}^2}{\text{d}f^2} - \frac{{\text{d}^2} \twon{Q\sbra{f}}^2}{\text{d}f^2} }} \\
&\leq \sbra{-2c_1 + \frac{8}{3}\pi^2\sbra{8C_1\epsilon + 4\epsilon^2}}n^2. \label{eq:Qbar2d} \end{split}}

Letting
\equ{\epsilon = \frac{1}{2}\min\lbra{c_2,\; \frac{3c_1}{4\pi^2\sbra{8C_1+4}}} \label{eq:epsilon}}
and substituting \eqref{eq:epsilon} into \eqref{eq:Qbarbound} and \eqref{eq:Qbar2d}, we have
\lentwo{\equa{\twon{\overline Q(f)}
    &\leq& 1-c'_2, \quad f\in \cF, \label{eq:Qbfar}\\ \frac{\text{d}^2\twon{\overline Q(f)}^2}{\text{d}f^2}
    &\leq& -2c'_1n^2,\quad f\in\cN_k, \label{eq:Qbnear1}
    }}where $c'_1$ and $c'_2$ are positive constants as well.
By consecutively applying  $\twon{\overline Q(f_k)}^2=1$ (according to \eqref{eq:Qfj}), $\frac{\text{d}\twon{\overline Q(f)}^2}{\text{d}f}\Big|_{f=f_k}=0$ (according to \eqref{eq:Qdfj1}) and \eqref{eq:Qbnear1}, we have for $f\in\cN_k$,
\equ{\begin{split} \twon{\overline Q(f)}^2
&= 1 + \int_{f_k}^f \frac{\text{d}\twon{\overline Q(s)}^2}{\text{d}s} \text{d}s \\
&= 1 + \int_{f_k}^f \text{d}s \int_{f_k}^s \frac{\text{d}^2\twon{\overline Q(t)}^2}{\text{d}t^2} \text{d}t \\
&\leq 1 + \int_{f_k}^f \text{d}s \int_{f_k}^s -2c'_1n^2 \text{d}t \\
&= 1 - c'_1n^2 (f-f_k)^2. \end{split} \label{eq:Qbnear}}
It follows from \eqref{eq:Qbfar} and \eqref{eq:Qbnear} that $\overline Q(f)$ satisfies \eqref{eq:cons2}. Therefore, $\overline Q(f)$ is a polynomial satisfying the constraints in \eqref{eq:cons1}--\eqref{eq:cons3}, as required in Proposition \ref{prop:dualpoly}. Finally, substituting \eqref{eq:epsilon} into \eqref{eq:Mboundcontinu} gives the bound on $M$ in \eqref{eq:AN_bound}, completing the proof.

\section{Conclusion} \label{sec:conclusion}
A rigorous analysis was performed in this paper to confirm the observation that the frequency estimation performance of atomic norm minimization improves as the number of channels increases. The sample size per channel that ensures exact frequency estimation from noiseless data was derived as a decreasing function of the number of channels. Numerical results were provided that agree with our analysis.

While we have shown the performance gain of the use of multichannel data in reducing the sample size per channel, future work is to answer the question as to whether the resolution can be improved by increasing the number of channels. A positive answer to the question was suggested by empirical evidence presented in \cite{yang2016exact,li2016off,fernandez2016super,yang2016enhancing}. Another interesting and important future work is to analyze the noise robustness of atomic norm minimization with compressive data.

\appendix

\subsection{Proof of Lemma \ref{lem:ymlogy}} \label{append:ymlogy}
Consider $x$ as a function of $y\geq 1$:
\equ{x(y) = y-\log y -1.\notag}
Since the derivative $x'(y) = 1-1/y>0$, as $y>1$, we have that $x(y)\geq 0$ is monotonically increasing on $\left[1,+\infty\right)$. As a result, $y(x)$ is the inverse function of $x(y)$ and is monotonically increasing.

To show the second part of the lemma, we first note that
\equ{y = x + \log y +1 \geq x+1.\notag}
Then, using the fact that $\log y \leq \frac{1}{2}y$ for any $y\geq 1$, we have
\equ{x = y-\log y -1 \geq y - \frac{1}{2}y - 1 = \frac{1}{2}y - 1.\notag}
It follows that $y \leq 2(1+x)$, completing the proof.

\subsection{Proof of Lemma \ref{lem:I1}} \label{append:I1}
The proof of this lemma follows similar steps as those of the proofs of \cite[Lemma IV.8]{tang2012compressed} and \cite[Lemma 6.6]{yang2014exact1}. A main difference is the use of Lemma \ref{lem:bernstein}, rather than Hoeffding's inequality.

Recall that $I_1^l\sbra{f}
= \inp{\m{\Phi}, \overline{\m{L}}^H\sbra{\overline{\m{v}}_l\sbra{f}-p\m{v}_l\sbra{f} }}$, where the rows of $\m{\Phi}$ are independent random vectors that are uniformly distributed on the unit sphere. Conditioned on a particular realization $\omega\in\cE$ where
\equ{\begin{split}\cE= \Bigg\{\omega:&\;
\sup_{f_d\in \bT_{\text{grid}}} \twon{\overline{\m{L}}^H\sbra{\overline{\m{v}}_l\sbra{f_d}-p\m{v}_l\sbra{f_d} }} <\lambda_l,\\
&\phantom{\twon{\overline{\m{L}}^H\sbra{\overline{\m{v}}_l\sbra{f_d}-p\m{v}_l\sbra{f_d} }}} l=0,1,2,3 \Bigg\}, \end{split} \notag}
Lemma \ref{lem:bernstein} and the union bound then imply, for $\epsilon>\lambda_l$,
\equ{\begin{split}
&\bP\sbra{\sup_{f_d\in \bT_{\text{grid}}} \twon{\inp{\m{\Phi},\;\overline{\m{L}}^H\sbra{\overline{\m{v}}_l\sbra{f_d}-p\m{v}_l\sbra{f_d} }}}> \epsilon \middle\vert \omega} \\
&\leq \abs{\bT_{\text{grid}}} e^{-L\sbra{\frac{\epsilon^2}{\lambda_l^2} - \log \frac{\epsilon^2}{\lambda_l^2} - 1}}. \label{eq:epsilonglambda} \end{split}}
It immediately follows that
\equ{\begin{split}
&\bP\sbra{\sup_{f_d\in \bT_{\text{grid}}} \twon{\inp{\m{\Phi},\;\overline{\m{L}}^H\sbra{\overline{\m{v}}_l\sbra{f_d}-p\m{v}_l\sbra{f_d} }}} > \epsilon} \\
&\leq \abs{\bT_{\text{grid}}} e^{-L\sbra{\frac{\epsilon^2}{\lambda_l^2} - \log \frac{\epsilon^2}{\lambda_l^2} - 1}} + \bP\sbra{\cE^c}. \end{split} \notag}
Setting
\equ{\lambda_l = 4\sbra{2^{2l+3}\sqrt{\frac{K}{M}}+\frac{n}{M}a\bar \sigma_l} \label{eq:lambdal}}
in $\cE$ and applying Lemma \ref{lem:I1bound} yields
\equ{\begin{split}
&\bP\sbra{\sup_{f_d\in \bT_{\text{grid}}} \twon{\inp{\m{\Phi},\;\overline{\m{L}}^H\sbra{\overline{\m{v}}_l\sbra{f_d}-p\m{v}_l\sbra{f_d} }}} > \epsilon} \\
&\leq \abs{\bT_{\text{grid}}} e^{-L\sbra{\frac{\epsilon^2}{\lambda_l^2} - \log \frac{\epsilon^2}{\lambda_l^2} - 1}} + 64\abs{\bT_{\text{grid}}}e^{-\gamma a^2}+\bP\sbra{\cE_{1,\tau}^c}. \label{eq:PI1} \end{split}}

For the second term to be no greater than $\delta$, $a$ is chosen such that
\equ{a^2 = \gamma^{-1}\log \frac{64\abs{\bT_{\text{grid}}}}{\delta} \notag}
and it will be fixed from now on. The first term is no greater than $\delta$ if
\equ{\frac{\epsilon^2}{\lambda_l^2} - \log\frac{\epsilon^2}{\lambda_l^2} - 1 \geq \frac{1}{L}\log \frac{\abs{\bT_{\text{grid}}}}{\delta}.\notag}
Applying Lemma \ref{lem:ymlogy}, this holds if
\equ{\frac{\epsilon^2}{\lambda_l^2} \geq 2 \sbra{1+\frac{1}{L}\log\frac{\abs{\bT_{\text{grid}}}}{\delta} }. \label{eq:ratioepstolam}}
Note that \eqref{eq:ratioepstolam} implies $\epsilon>\lambda_l$, which is required to verify \eqref{eq:epsilonglambda}.

We next derive a bound for $M$ given \eqref{eq:ratioepstolam}. First consider the case of $2^4K/\sqrt{M}\geq 1$. The condition in Lemma \ref{lem:I1bound} is $a\leq \sqrt{2}M^{1/4}$ or equivalently
\equ{M\geq \frac{1}{4}a^4 = \frac{1}{4}\gamma^{-2}\log^2 \frac{64\abs{\bT_{\text{grid}}}}{\delta}. \label{eq:Ma}}
In this case, we have $a\bar{\sigma}_l\leq 2^{2l+3}\frac{\sqrt{MK}}{n}$, which inserting into \eqref{eq:lambdal} results in
\equ{\frac{1}{\lambda_l^2} = \frac{1}{16\sbra{2^{2l+3}\sqrt{\frac{K}{M}}+\frac{n}{M}a\bar \sigma_l}^2} \geq \frac{1}{4^{2l+5}}\frac{M}{K}.\notag}
Now consider the other case of $2^4K/\sqrt{M}<1$, which will be discussed in two scenarios. If $32s\geq a^2$, then $a\bar{\sigma}_l\leq 2^{2l+3}\frac{\sqrt{MK}}{n}$ which again gives the above lower bound on $\frac{1}{\lambda_l^2}$. Otherwise if $32s \leq a^2$, then $\lambda_l \leq 2^{2l+3}\sqrt{2}\frac{a}{\sqrt{M}}$ and
\equ{\frac{1}{\lambda_l^2} \geq \frac{1}{2^{4l+7}} \frac{M}{a^2}.\notag}
Therefore, to make \eqref{eq:ratioepstolam} hold true, it suffices to take $M$ satisfying \eqref{eq:Ma} and
\equ{M\min\sbra{\frac{1}{4^{2l+5}}\frac{1}{K}, \frac{1}{2^{4l+7}} \frac{1}{a^2}} \geq \frac{2}{\epsilon^2}\sbra{1+\frac{1}{L}\log\frac{\abs{\bT_{\text{grid}}}}{\delta} }.\notag}
By the arguments above, the first term on the right hand side of \eqref{eq:PI1} is no greater than $\delta$ if
\equ{\begin{split} M\geq \max\bigg\{
&\frac{2}{\epsilon^2}4^{2l+5} K \sbra{1+\frac{1}{L}\log\frac{\abs{\bT_{\text{grid}}}}{\delta}},\\
&\frac{2}{\epsilon^2} 2^{4l+7}\gamma^{-1}\log\frac{64\abs{\bT_{\text{grid}}}}{\delta}  \sbra{1+\frac{1}{L}\log\frac{\abs{\bT_{\text{grid}}}}{\delta}},\\
&\frac{1}{4}\gamma^{-2}\log^2 \frac{64\abs{\bT_{\text{grid}}}}{\delta} \bigg\}. \label{eq:MI1_1} \end{split}}
According to Lemma \ref{lem:invertibility}, the last term on the right hand side of \eqref{eq:PI1} is no greater than $\delta$ if
\equ{M\geq \frac{50}{\tau^2}K\log \frac{2K}{\delta}. \label{eq:MI1_2}}

Setting $\tau=\frac{1}{4}$, combining \eqref{eq:MI1_1} and \eqref{eq:MI1_2} together, absorbing all constants into one, and using the inequality $\log\frac{\abs{\bT_{\text{grid}}}}{\delta}  \sbra{1+\frac{1}{L}\log\frac{\abs{\bT_{\text{grid}}}}{\delta}}\leq 2\log^2\frac{\abs{\bT_{\text{grid}}}}{\delta}$, we have
\equ{\begin{split}M\geq C\max\bigg\{
&\frac{1}{\epsilon^2}K \sbra{1+\frac{1}{L}\log\frac{\abs{\bT_{\text{grid}}}}{\delta}},\\
&\frac{1}{\epsilon^2} \log^2 \frac{\abs{\bT_{\text{grid}}}}{\delta},\; K\log \frac{K}{\delta}\bigg\} \end{split}\notag}
is sufficient to guarantee
\equ{\sup_{f_d\in\bT_{\text{grid}}} \twon{I_1^l\sbra{f_d}}\leq \epsilon\notag}
with probability at least $1-3\delta$. Applying the union bound then completes the proof.

\subsection{Proof of Lemma \ref{lem:I2}} \label{append:I2}
Recall that $I_2^l\sbra{f}
= \inp{\m{\Phi}, \sbra{\overline{\m{L}}-p^{-1}\m{L}}^Hp\m{v}_l\sbra{f}}$. According to Lemma \ref{lem:I2bound}, we have on the set $\cE_{1,\tau}$
\equ{\twon{\sbra{\overline{\m{L}}-p^{-1}\m{L} }^Hp\m{v}_l\sbra{f} } \leq C\tau \notag}
for some constant $C>0$. Applying Lemma \ref{lem:bernstein} and the union bound gives for $\epsilon > C\tau$,
\equ{\begin{split}
&\bP\sbra{\sup_{f_d\in\bT_{\text{grid}}} \twon{I_2^l\sbra{f_d}} >\epsilon} \\
&\leq \abs{\bT_{\text{grid}}} e^{-L\sbra{\frac{\epsilon^2}{C\tau^2} - \log\frac{\epsilon^2}{C\tau^2} -1 }} + \bP\sbra{\cE_{1,\tau}^c}.\end{split}\notag}
To make the first term no greater than $\delta$, we take $\tau$ such that
\equ{\frac{\epsilon^2}{C\tau^2} - \log\frac{\epsilon^2}{C\tau^2} -1 \geq \frac{1}{L}\log \frac{\abs{\bT_{\text{grid}}}}{\delta}.\notag}
According to Lemma \ref{lem:ymlogy}, it then suffices to fix $\tau$ such that
\equ{\frac{1}{\tau^2} = C\frac{2}{\epsilon^2}\sbra{1+ \frac{1}{L}\log \frac{\abs{\bT_{\text{grid}}}}{\delta}}.\notag}
To make the second term no greater than $\delta$, it suffices by Lemma \ref{lem:invertibility} to take
\equ{\begin{split}M
&\geq \frac{C}{\tau^2}K\log \frac{2K}{\delta} \\
&= C\frac{1}{\epsilon^2}K\log \frac{2K}{\delta}\sbra{1+ \frac{1}{L}\log \frac{\abs{\bT_{\text{grid}}}}{\delta}}. \end{split}\notag}
Application of the union bound proves the lemma.

\subsection{Proof of Lemma \ref{lem:extendcontin}} \label{append:extendcontin}
Recall by inserting \eqref{eq:K} into $\m{v}_l(f)$ as given in \eqref{eq:vl} that
\equ{\m{v}_l(f) = \frac{1}{n}\sum_{j=-2n}^{j=2n} \sbra{\frac{i2\pi j}{c_0}}^l g_n(j) e^{i2\pi f j}\m{e}(j)\notag}
where
\equ{\m{e}(j) = \begin{bmatrix} e^{-i2\pi f_1 j} \\ \vdots \\ e^{-i2\pi f_K j} \\ \frac{i2\pi j}{c_0^2}e^{-i2\pi f_1 j} \\ \vdots \\ \frac{i2\pi j}{c_0^2}e^{-i2\pi f_K j}\end{bmatrix}.\notag}
It follows that
\equ{\begin{split}
&\m{v}_l\sbra{f} - \m{v}_l\sbra{f_d} \\
&= \frac{1}{n}\sum_{j=-2n}^{j=2n} \sbra{\frac{i2\pi j}{c_0}}^l g_n(j)\sbra{e^{i2\pi f j} - e^{i2\pi f_d j}}\m{e}(j) . \end{split}\notag}
Using the following bounds shown in \cite{candes2013towards,tang2012compressed}:
{\lentwo\equa{ \norm{g_n}_{\infty}
&\leq& 1, \notag\\ \abs{\frac{i2\pi j}{c_0}}
&\leq& 4 \text{ when } n\geq 2, \notag\\ \twon{\m{e}(j)}^2
&\leq& 14K \text{ when } n\geq 4
\notag}}and the bound
\equ{\begin{split}\abs{e^{i2\pi f j} - e^{i2\pi f_d j}}
&= \abs{e^{i\pi (f+f_d) j}\cdot 2i\sin\sbra{\pi(f-f_d)j} } \\
&= 2\abs{\sin\sbra{\pi(f-f_d)j}} \\
&\leq 2\pi\abs{f-f_d}\cdot\abs{j} \\
&\leq 4n\pi\abs{f-f_d}, \end{split}\notag}
we have
\equ{\begin{split}
&\twon{\m{v}_l\sbra{f} - \m{v}_l\sbra{f_d}}\\
&\leq \frac{1}{n}\sum_{j=-2n}^{j=2n} \abs{\frac{i2\pi j}{c_0}}^l \abs{g_n(j)} \abs{e^{i2\pi f j} - e^{i2\pi f_d j}}\twon{\m{e}(j)} \\
&\leq\frac{1}{n}\cdot (4n+1)\cdot 4^l \cdot 4n\pi\abs{f-f_d} \cdot 14K \\
&\leq CnK\abs{f-f_d}\\
&\leq Cn^2\abs{f-f_d} . \end{split}\notag}
We then obtain
\equ{\begin{split}\twon{\m{L}^H\sbra{\m{v}_l\sbra{f} - \m{v}_l\sbra{f_d}}}
&\leq \twon{\m{L}} \twon{\m{v}_l\sbra{f} - \m{v}_l\sbra{f_d}}\\
&\leq Cn^2\abs{f-f_d}, \end{split}\notag}
since $\twon{\m{L}} \leq \twon{\m{D}^{-1}} \leq 1.568$ according to \cite{candes2013towards}.

The bound in \eqref{eq:Lvrand} can be shown using similar arguments, while the exponent for $n$ increases from 2 to 3 since $\twon{\overline{\m{L}}} \leq 2\twon{\m{D}^{-1}}p^{-1} \leq C\frac{n}{M} \leq Cn$ according to \cite{tang2012compressed}.

\subsection{Proof of Lemma \ref{lem:Qfdf}} \label{append:Qfdf}

It follows from Lemma \ref{lem:extendcontin} that
\equ{\begin{split}
&\twon{\overline{\m{L}}^H \sbra{\overline{\m{v}}_l\sbra{f_d} - \overline{\m{v}}_l\sbra{f}} + \m{L}^H \sbra{\m{v}_l\sbra{f_d} - \m{v}_l\sbra{f}}} \\
&\leq \twon{\overline{\m{L}}^H \sbra{\overline{\m{v}}_l\sbra{f_d} - \overline{\m{v}}_l\sbra{f}}} + \twon{\m{L}^H \sbra{\m{v}_l\sbra{f_d} - \m{v}_l\sbra{f}}} \\
&\leq Cn^3\abs{f-f_d}. \end{split}\notag}
We then recall \eqref{eq:Qdc} and apply Lemma \ref{lem:bernstein}, having for $\epsilon > Cn^3\abs{f-f_d}$,
\equ{\begin{split}
&\bP\Big(c_0^{-l}\twon{\sbra{\overline{Q}^{\sbra{l}}\sbra{f_d} - \overline{Q}^{\sbra{l}}\sbra{f}} + \sbra{Q^{\sbra{l}}\sbra{f_d} - Q^{\sbra{l}}\sbra{f}}} \\ &\phantom{c\twon{\sbra{\overline{Q}^{\sbra{l}}\sbra{f_d} - \overline{Q}^{\sbra{l}}\sbra{f}} + \sbra{Q^{\sbra{l}}\sbra{f_d} - Q^{\sbra{l}}\sbra{f}}}}\geq \epsilon\Big) \\
&\leq e^{-L\sbra{\frac{\epsilon^2}{Cn^6\abs{f-f_d}^2} - \log\frac{\epsilon^2}{Cn^6\abs{f-f_d}^2} - 1}} + \bP\sbra{\cE_{1,\tau}^c}. \end{split}\notag}
For the second term to be no greater than $\delta$, it suffices to take $M\geq CK\log \frac{K}{\delta}$ by Lemma \ref{lem:invertibility}. For the first term to be no greater than $\delta$, according to Lemma \ref{lem:ymlogy}, it suffices to let
\equ{\frac{\epsilon^2}{Cn^6\abs{f-f_d}^2} \geq 2\sbra{1+\frac{1}{L}\log \frac{1}{\delta}},\notag}
or equivalently
\equ{\frac{1}{\abs{f-f_d}} \geq C\frac{n^3}{\epsilon}\sqrt{1+\frac{1}{L}\log \frac{1}{\delta}}. \label{fmfdinv}}
Therefore, to guarantee, for any $f\in\bT$, that some $f_d\in \bT_{\text{grid}}$ can always be found such that \eqref{fmfdinv} holds, it suffices to let $\bT_{\text{grid}}$ be a uniform grid with
\equ{\abs{\bT_{\text{grid}}} = \left\lceil C\frac{n^3}{\epsilon}\sqrt{1+\frac{1}{L}\log \frac{1}{\delta}} \right\rceil.\notag}
Application of the union bound completes the proof.

\subsection{Proof of Corollary \ref{cor:l21}} \label{append:corl21}
Similar to atomic norm minimization in Theorem \ref{thm:noiseless}, to certify optimality for $\ell_{2,1}$ norm minimization in \eqref{eq:l21} it suffices to construct a dual certificate $\overline{Q}\sbra{f} = \m{a}^H\sbra{f}\m{V}$ satisfying (see also \cite[Theorem 3.1]{eldar2010average})
\lentwo{\equa{\overline  Q\sbra{f_k}
&=& \m{\phi}_k, \quad f_k\in\cT, \label{eq:dcons1}\\ \twon{\overline Q\sbra{f}}
&<& 1, \quad f\in\lbra{\widetilde{f}_g}_{g=1}^G\backslash\cT, \label{eq:dcons2}\\ \m{V}_j
&=& \m{0}, \quad j\in \Omega^c. \label{eq:dcons3}
}}The only difference from the dual certificate for atomic norm minimization lies in that only finitely many constraints are involved in \eqref{eq:dcons2}. Therefore, the dual certificate constructed in Section \ref{sec:proof} for atomic norm minimization is naturally a certificate for $\ell_{2,1}$ norm minimization.


\begin{thebibliography}{10}
\providecommand{\url}[1]{#1}
\csname url@samestyle\endcsname
\providecommand{\newblock}{\relax}
\providecommand{\bibinfo}[2]{#2}
\providecommand{\BIBentrySTDinterwordspacing}{\spaceskip=0pt\relax}
\providecommand{\BIBentryALTinterwordstretchfactor}{4}
\providecommand{\BIBentryALTinterwordspacing}{\spaceskip=\fontdimen2\font plus
\BIBentryALTinterwordstretchfactor\fontdimen3\font minus
  \fontdimen4\font\relax}
\providecommand{\BIBforeignlanguage}[2]{{%
\expandafter\ifx\csname l@#1\endcsname\relax
\typeout{** WARNING: IEEEtran.bst: No hyphenation pattern has been}%
\typeout{** loaded for the language `#1'. Using the pattern for}%
\typeout{** the default language instead.}%
\else
\language=\csname l@#1\endcsname
\fi
#2}}
\providecommand{\BIBdecl}{\relax}
\BIBdecl

\bibitem{yang2018average}
Z.~Yang, Y.~C. Eldar, and L.~Xie, ``Average case analysis of compressive
  multichannel frequency estimation using atomic norm minimization,'' in
  \emph{23rd International Conference on Digital Signal Processing (DSP)},
  2018.

\bibitem{stoica2005spectral}
P.~Stoica and R.~L. Moses, \emph{Spectral analysis of signals}.\hskip 1em plus
  0.5em minus 0.4em\relax Pearson/Prentice Hall Upper Saddle River, NJ, 2005.

\bibitem{tang2012compressed}
G.~Tang, B.~N. Bhaskar, P.~Shah, and B.~Recht, ``Compressed sensing off the
  grid,'' \emph{IEEE Transactions on Information Theory}, vol.~59, no.~11, pp.
  7465--7490, 2013.

\bibitem{candes2013towards}
E.~J. Cand{\`e}s and C.~Fernandez-Granda, ``Towards a mathematical theory of
  super-resolution,'' \emph{Communications on Pure and Applied Mathematics},
  vol.~67, no.~6, pp. 906--956, 2014.

\bibitem{krim1996two}
H.~Krim and M.~Viberg, ``Two decades of array signal processing research: The
  parametric approach,'' \emph{IEEE Signal Processing Magazine}, vol.~13,
  no.~4, pp. 67--94, 1996.

\bibitem{heylen2006modal}
W.~Heylen and P.~Sas, \emph{Modal analysis theory and testing}.\hskip 1em plus
  0.5em minus 0.4em\relax Katholieke Universteit Leuven, 2006.

\bibitem{barbotin2012estimation}
Y.~Barbotin, A.~Hormati, S.~Rangan, and M.~Vetterli, ``{Estimation of sparse
  MIMO channels with common support},'' \emph{IEEE Transactions on
  Communications}, vol.~60, no.~12, pp. 3705--3716, 2012.

\bibitem{li2007mimo}
J.~Li and P.~Stoica, ``{MIMO radar with colocated antennas},'' \emph{IEEE
  Signal Processing Magazine}, vol.~24, no.~5, pp. 106--114, 2007.

\bibitem{baransky2014sub}
E.~Baransky, G.~Itzhak, N.~Wagner, I.~Shmuel, E.~Shoshan, and Y.~C. Eldar,
  ``{Sub-Nyquist radar prototype: Hardware and algorithm},'' \emph{IEEE
  Transactions on Aerospace and Electronic Systems}, vol.~50, no.~2, pp.
  809--822, 2014.

\bibitem{eldar2015sampling}
Y.~C. Eldar, \emph{Sampling Theory: Beyond Bandlimited Systems}.\hskip 1em plus
  0.5em minus 0.4em\relax Cambridge University Press, 2015.

\bibitem{rust2006sub}
M.~J. Rust, M.~Bates, and X.~Zhuang, ``Sub-diffraction-limit imaging by
  stochastic optical reconstruction microscopy (storm),'' \emph{Nature
  Methods}, vol.~3, no.~10, pp. 793--796, 2006.

\bibitem{rossi2014spatial}
M.~Rossi, A.~M. Haimovich, and Y.~C. Eldar, ``{Spatial compressive sensing for
  MIMO radar},'' \emph{IEEE Transactions on Signal Processing}, vol.~62, no.~2,
  pp. 419--430, 2014.

\bibitem{schervish2012theory}
M.~J. Schervish, \emph{Theory of statistics}.\hskip 1em plus 0.5em minus
  0.4em\relax Springer Science \& Business Media, 2012.

\bibitem{wax1989unique}
M.~Wax and I.~Ziskind, ``On unique localization of multiple sources by passive
  sensor arrays,'' \emph{IEEE Transactions on Acoustics, Speech, and Signal
  Processing}, vol.~37, no.~7, pp. 996--1000, 1989.

\bibitem{candes2013super}
E.~J. Cand{\`e}s and C.~Fernandez-Granda, ``Super-resolution from noisy data,''
  \emph{Journal of Fourier Analysis and Applications}, vol.~19, no.~6, pp.
  1229--1254, 2013.

\bibitem{azais2015spike}
J.-M. Azais, Y.~De~Castro, and F.~Gamboa, ``Spike detection from inaccurate
  samplings,'' \emph{Applied and Computational Harmonic Analysis}, vol.~38,
  no.~2, pp. 177--195, 2015.

\bibitem{tang2015near}
G.~Tang, B.~N. Bhaskar, and B.~Recht, ``Near minimax line spectral
  estimation,'' \emph{IEEE Transactions on Information Theory}, vol.~61, no.~1,
  pp. 499--512, 2015.

\bibitem{yang2016exact}
Z.~Yang and L.~Xie, ``Exact joint sparse frequency recovery via optimization
  methods,'' \emph{IEEE Transactions on Signal Processing}, vol.~64, no.~19,
  pp. 5145--5157, 2016.

\bibitem{li2016off}
Y.~Li and Y.~Chi, ``Off-the-grid line spectrum denoising and estimation with
  multiple measurement vectors,'' \emph{IEEE Transactions on Signal
  Processing}, vol.~64, no.~5, pp. 1257--1269, 2016.

\bibitem{fernandez2017demixing}
C.~Fernandez-Granda, G.~Tang, X.~Wang, and L.~Zheng, ``Demixing sines and
  spikes: Robust spectral super-resolution in the presence of outliers,''
  \emph{Information and Inference: A Journal of the IMA}, p. iax005.

\bibitem{li2018atomic}
S.~Li, D.~Yang, G.~Tang, and M.~B. Wakin, ``Atomic norm minimization for modal
  analysis from random and compressed samples,'' \emph{IEEE Transactions on
  Signal Processing}, vol.~66, no.~7, pp. 1817--1831, 2018.

\bibitem{fernandez2016super}
C.~Fernandez-Granda, ``Super-resolution of point sources via convex
  programming,'' \emph{Information and Inference: A Journal of the IMA},
  vol.~5, no.~3, pp. 251--303, 2016.

\bibitem{gribonval2008atoms}
R.~Gribonval, H.~Rauhut, K.~Schnass, and P.~Vandergheynst, ``{Atoms of all
  channels, unite! Average case analysis of multi-channel sparse recovery using
  greedy algorithms},'' \emph{Journal of Fourier Analysis and Applications},
  vol.~14, no. 5-6, pp. 655--687, 2008.

\bibitem{eldar2010average}
Y.~C. Eldar and H.~Rauhut, ``Average case analysis of multichannel sparse
  recovery using convex relaxation,'' \emph{IEEE Transactions on Information
  Theory}, vol.~56, no.~1, pp. 505--519, 2010.

\bibitem{yang2016enhancing}
Z.~Yang and L.~Xie, ``Enhancing sparsity and resolution via reweighted atomic
  norm minimization,'' \emph{IEEE Transactions on Signal Processing}, vol.~64,
  no.~4, pp. 995--1006, 2016.

\bibitem{malioutov2005sparse}
D.~Malioutov, M.~Cetin, and A.~S. Willsky, ``A sparse signal reconstruction
  perspective for source localization with sensor arrays,'' \emph{IEEE
  Transactions on Signal Processing}, vol.~53, no.~8, pp. 3010--3022, 2005.

\bibitem{stoica1989music}
P.~Stoica and N.~Arye, ``{MUSIC, maximum likelihood, and Cramer-Rao bound},''
  \emph{IEEE Transactions on Acoustics, Speech and Signal Processing}, vol.~37,
  no.~5, pp. 720--741, 1989.

\bibitem{schmidt1981signal}
R.~Schmidt, ``A signal subspace approach to multiple emitter location spectral
  estimation,'' Ph.D. dissertation, Stanford University, 1981.

\bibitem{roy1989esprit}
R.~Roy and T.~Kailath, ``{ESPRIT}-estimation of signal parameters via
  rotational invariance techniques,'' \emph{IEEE Transactions on Acoustics,
  Speech and Signal Processing}, vol.~37, no.~7, pp. 984--995, 1989.

\bibitem{liao2016music}
W.~Liao and A.~Fannjiang, ``{MUSIC} for single-snapshot spectral estimation:
  Stability and super-resolution,'' \emph{Applied and Computational Harmonic
  Analysis}, vol.~40, no.~1, pp. 33--67, 2016.

\bibitem{fannjiang2016compressive}
A.~Fannjiang, ``Compressive spectral estimation with single-snapshot {ESPRIT}:
  Stability and resolution,'' \emph{arXiv preprint arXiv:1607.01827}, 2016.

\bibitem{feng1996spectrum}
P.~Feng and Y.~Bresler, ``Spectrum-blind minimum-rate sampling and
  reconstruction of multiband signals,'' in \emph{IEEE International Conference
  on Acoustics, Speech, and Signal Processing (ICASSP)}, vol.~3, 1996, pp.
  1688--1691.

\bibitem{davies2012rank}
M.~E. Davies and Y.~C. Eldar, ``Rank awareness in joint sparse recovery,''
  \emph{IEEE Transactions on Information Theory}, vol.~58, no.~2, pp.
  1135--1146, 2012.

\bibitem{kim2012compressive}
J.~M. Kim, O.~K. Lee, and J.~C. Ye, ``{Compressive MUSIC: Revisiting the link
  between compressive sensing and array signal processing},'' \emph{IEEE
  Transactions on Information Theory}, vol.~58, no.~1, pp. 278--301, 2012.

\bibitem{lee2012subspace}
K.~Lee, Y.~Bresler, and M.~Junge, ``Subspace methods for joint sparse
  recovery,'' \emph{IEEE Transactions on Information Theory}, vol.~58, no.~6,
  pp. 3613--3641, 2012.

\bibitem{eldar2012compressed}
Y.~C. Eldar and G.~Kutyniok, \emph{Compressed sensing: theory and
  applications}.\hskip 1em plus 0.5em minus 0.4em\relax Cambridge University
  Press, 2012.

\bibitem{donoho2001uncertainty}
D.~L. Donoho and X.~Huo, ``Uncertainty principles and ideal atomic
  decomposition,'' \emph{IEEE Transactions on Information Theory}, vol.~47,
  no.~7, pp. 2845--2862, 2001.

\bibitem{candes2006compressive}
E.~Cand{\`e}s, ``Compressive sampling,'' in \emph{Proceedings of the
  International Congress of Mathematicians}, vol.~3, 2006, pp. 1433--1452.

\bibitem{pati1993orthogonal}
Y.~C. Pati, R.~Rezaiifar, and P.~Krishnaprasad, ``Orthogonal matching pursuit:
  Recursive function approximation with applications to wavelet
  decomposition,'' in \emph{1993 Conference Record of The Twenty-Seventh
  Asilomar Conference on Signals, Systems and Computers}, 1993, pp. 40--44.

\bibitem{tropp2006algorithms}
J.~A. Tropp, A.~C. Gilbert, and M.~J. Strauss, ``{Algorithms for simultaneous
  sparse approximation. Part I: Greedy pursuit},'' \emph{Signal Processing},
  vol.~86, no.~3, pp. 572--588, 2006.

\bibitem{gorodnitsky1997sparse}
I.~F. Gorodnitsky and B.~D. Rao, ``{Sparse signal reconstruction from limited
  data using FOCUSS: A re-weighted minimum norm algorithm},'' \emph{IEEE
  Transactions on Signal Processing}, vol.~45, no.~3, pp. 600--616, 1997.

\bibitem{stoica2011spice}
P.~Stoica, P.~Babu, and J.~Li, ``{SPICE}: A sparse covariance-based estimation
  method for array processing,'' \emph{IEEE Transactions on Signal Processing},
  vol.~59, no.~2, pp. 629--638, 2011.

\bibitem{candes2006robust}
E.~J. Cand{\`e}s, J.~Romberg, and T.~Tao, ``{Robust uncertainty principles:
  Exact signal reconstruction from highly incomplete frequency information},''
  \emph{IEEE Transactions on Information Theory}, vol.~52, no.~2, pp. 489--509,
  2006.

\bibitem{haviv2017restricted}
I.~Haviv and O.~Regev, ``{The restricted isometry property of subsampled
  Fourier matrices},'' in \emph{Geometric Aspects of Functional
  Analysis}.\hskip 1em plus 0.5em minus 0.4em\relax Springer, 2017, pp.
  163--179.

\bibitem{chandrasekaran2012convex}
V.~Chandrasekaran, B.~Recht, P.~A. Parrilo, and A.~S. Willsky, ``The convex
  geometry of linear inverse problems,'' \emph{Foundations of Computational
  Mathematics}, vol.~12, no.~6, pp. 805--849, 2012.

\bibitem{yang2014discretization}
Z.~Yang, L.~Xie, and C.~Zhang, ``A discretization-free sparse and parametric
  approach for linear array signal processing,'' \emph{IEEE Transactions on
  Signal Processing}, vol.~62, no.~19, pp. 4959--4973, 2014.

\bibitem{yang2015gridless}
Z.~Yang and L.~Xie, ``On gridless sparse methods for line spectral estimation
  from complete and incomplete data,'' \emph{IEEE Transactions on Signal
  Processing}, vol.~63, no.~12, pp. 3139--3153, 2015.

\bibitem{yang2018sparse}
Z.~Yang, J.~Li, P.~Stoica, and L.~Xie, ``Sparse methods for
  direction-of-arrival estimation,'' \emph{Academic Press Library in Signal
  Processing Volume 7 (R. Chellappa and S. Theodoridis, {\em Eds.})}, pp.
  509--581, 2018.

\bibitem{fang2016super}
J.~Fang, F.~Wang, Y.~Shen, H.~Li, and R.~Blum, ``Super-resolution compressed
  sensing for line spectral estimation: An iterative reweighted approach,''
  \emph{IEEE Transactions on Signal Processing}, vol.~64, no.~18, pp.
  4649--4662, 2016.

\bibitem{wu2018two}
X.~Wu, W.-P. Zhu, J.~Yan, and Z.~Zhang, ``Two sparse-based methods for off-grid
  direction-of-arrival estimation,'' \emph{Signal Processing}, vol. 142, pp.
  87--95, 2018.

\bibitem{steffens2018compact}
C.~Steffens, M.~Pesavento, and M.~E. Pfetsch, ``A compact formulation for the
  $\ell_{2, 1}$ mixed-norm minimization problem,'' \emph{IEEE Transactions on
  Signal Processing}, vol.~66, no.~6, pp. 1483--1497, 2018.

\bibitem{dumitrescu2007positive}
B.~Dumitrescu, \emph{Positive trigonometric polynomials and signal processing
  applications}.\hskip 1em plus 0.5em minus 0.4em\relax Springer, 2007.

\bibitem{haghighatshoar2017massive}
S.~Haghighatshoar and G.~Caire, ``{Massive MIMO channel subspace estimation
  from low-dimensional projections},'' \emph{IEEE Transactions on Signal
  Processing}, vol.~65, no.~2, pp. 303--318, 2017.

\bibitem{grant2008cvx}
M.~Grant and S.~Boyd, ``{CVX: Matlab software for disciplined convex
  programming},'' \emph{Available online at http://cvxr.com/cvx}, 2008.

\bibitem{tropp2008conditioning}
J.~A. Tropp, ``On the conditioning of random subdictionaries,'' \emph{Applied
  and Computational Harmonic Analysis}, vol.~25, no.~1, pp. 1--24, 2008.

\bibitem{yang2014exact1}
\BIBentryALTinterwordspacing
Z.~Yang and L.~Xie, ``Exact joint sparse frequency recovery via optimization
  methods,'' Tech. Rep., May 2014. [Online]. Available:
  \url{http://arxiv.org/abs/1405.6585v1}
\BIBentrySTDinterwordspacing

\end{thebibliography}

\end{document}